\documentclass[journal]{IEEEtran}

\usepackage{graphicx}
\usepackage[colorlinks,bookmarksopen,bookmarksnumbered,citecolor=blue,urlcolor=red]{hyperref}
\usepackage{amssymb}
\usepackage{amsbsy}
\usepackage{ulem}
\usepackage{upgreek}
\usepackage{algorithm}
\usepackage{algorithmicx}
\usepackage{algpseudocode}
\usepackage{threeparttable}
\usepackage[cmex10]{amsmath}
\usepackage{array}
\usepackage[tight,footnotesize]{subfigure}
\usepackage{caption}
\usepackage{url}
\usepackage{epstopdf}
\usepackage{verbatim}
\usepackage{color}
\usepackage{float}
\usepackage{wrapfig}
\usepackage{xspace}

\newcommand*{\ie}{i.e.\@\xspace}
\usepackage{enumerate}

\newcommand*{\etc}{
    \@ifnextchar{.}
        {etc}
        {etc.\@\xspace}
}
\graphicspath{{figures/}}

\begin{document}

\newtheorem{thm}{Theorem}
\newtheorem{lma}{Lemma}
\newtheorem{defi}{Definition}
\newtheorem{proper}{Property}

\title{Optimizing Resource Allocation and VNF Embedding in RAN Slicing}

\author{Tu N. Nguyen, Kashyab J. Ambarani, and My T. Thai

\IEEEcompsocitemizethanks{\IEEEcompsocthanksitem The short version of the paper has been accepted to the 2021 IEEE Global Communications Conference. Corresponding author: Tu N. Nguyen}

\IEEEcompsocitemizethanks{\IEEEcompsocthanksitem Tu N. Nguyen is with the Department of Computer Science, Kennesaw State University, Marietta, GA 30060 USA (e-mail: tu.nguyen@kennesaw.edu).}

\IEEEcompsocitemizethanks{\IEEEcompsocthanksitem K. J. Ambarani is with the Department of Computer Science, Purdue University Fort Wayne, Fort Wayne, IN 46805, USA (e-mail: ambakj01@pfw.edu).}

\IEEEcompsocitemizethanks{\IEEEcompsocthanksitem M. T. Thai is with the Computer and Information Science and Engineering, University of Florida, Gainesville, FL 32612 USA (e-mail: mythai@cise.ufl.edu).}

}

\IEEEcompsoctitleabstractindextext{%
\begin{abstract}
5G radio access network (RAN) with \textit{network slicing methodology} plays a key role in the development of the next-generation network system.
RAN slicing focuses on splitting the substrate's resources into a set of self-contained programmable \textit{RAN slices}.
Leveraged by network function virtualization (NFV),
a RAN slice is constituted by various virtual network functions (VNFs) and virtual links that are embedded as instances on substrate nodes.
In this work, we focus on the following fundamental tasks: i) establishing the theoretical foundation for constructing a VNF mapping plan for RAN slice recovery optimization and ii) developing algorithms needed to map/embed VNFs efficiently.
In particular, we propose four efficient algorithms, including Resource-based Algorithm (RBA), Connectivity-based Algorithm (CBA), Group-based Algorithm (GBA), and Group-Connectivity-based Algorithm (GCBA) to solve the resource allocation and VNF mapping problem.
Extensive experiments are also conducted to validate the robustness
of RAN slicing via the proposed algorithms.

\end{abstract}
}

\maketitle
\IEEEdisplaynotcompsoctitleabstractindextext
\IEEEpeerreviewmaketitle

\section{Introduction} \label{sec:intro}

A new network technology is critical to leverage high demands for diverse vertical industry applications growing rapidly.
In recent years, Radio Access Network (RAN) slicing has become one of the most promising architectural technologies for the forthcoming 5G era \cite{10.1145/3241539.3241567,10.1145/3117811.3117831}.
RAN slicing completely overturns the traditional model of a single
ownership of all network resources and brings a new vision where the physical infrastructure resources are shared across many {\it RAN slices}. Each slice built on top of the underlying physical RAN (substrate) is a separate logical mobile network, which delivers a set of services with similar characteristics and is isolated from others \cite{7926923,7926922,8370043}.
Leveraged by network function virtualization (NFV), a RAN slice is constituted by various virtual network functions (VNFs) and {\it virtual links} that are embedded as instances on substrate nodes \cite{8004168,10.1145/3098822.3098826}.
RAN enforcement mechanisms enable a highly efficient resource management service and maximizes the resources configured \cite{5290389}. Efficient resource allocation and VNF embedding serve as one of the critical aspects in RAN slicing technologies \cite{5719526,6023087}. The resource allocation and VNF embedding problem is considered under a variety of constraints \cite{5290389,5705791}.

{\bf Motivation.} Due to the complexity of the RAN slicing, existing configuration schemes\footnote{Some other works refer to them as RAN enforcement.} in RAN slicing just focus on {\it resource allocation} \cite{8717789,7946928,6449268}. However, a RAN enforcement problem, in fact, cannot be addressed by only considering the resource allocation but also needs to consider the relation between the substrate nodes.
The relations between VNFs decide the ways in which the VNFs will be mapped into the substrate nodes.
In addition, because of the connection between VNFs, then bandwidth requirement of the virtual links between VNFs \cite{7524565,7410276} needs to be considered as well.
None of existing work fully considers above factors nor do they attempt to consider the interdependency property to mitigate the impact of the configuration process on the entire network.

\begin{figure*}[t]
\center
\subfigure{\includegraphics[width=15cm]{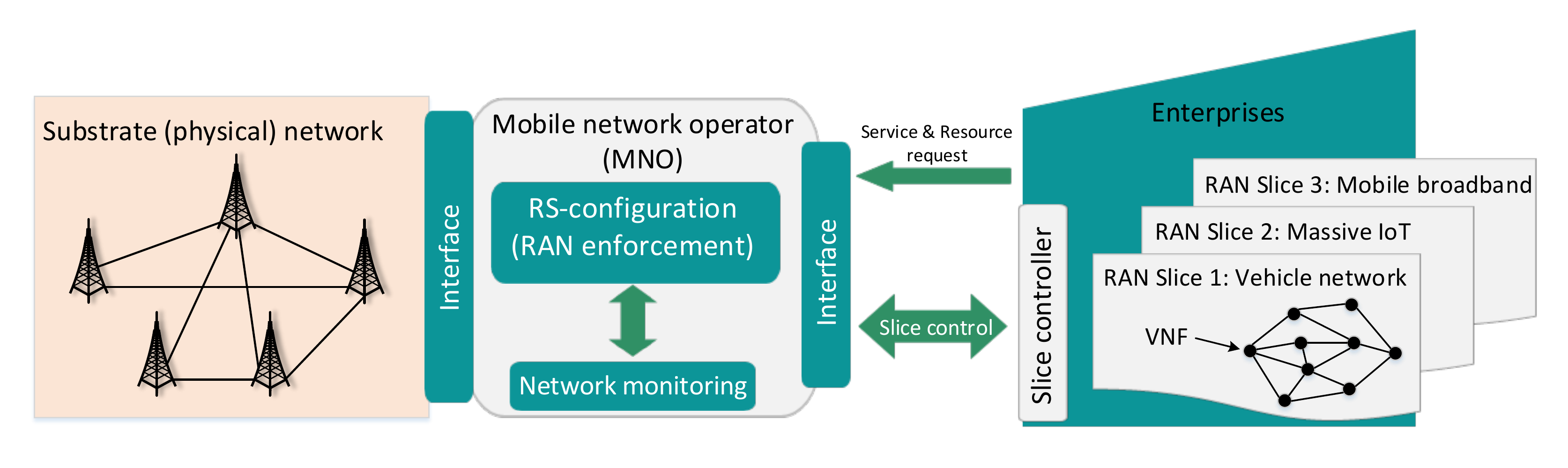}}
\caption{A high-level architecture for RAN slicing.}
\label{fig:ns}
\end{figure*}

\textbf{Contributions.}
In this work we propose a new RAN enforcement, namely {\it RAN slicing-configuration (RS-configuration)} that overcomes the drawbacks of the existing works to provide an optimal mapping plan for VNFs onto the substrate network. In particular, the RS-configuration will provide an ordered set of mappings of VNFs onto substrate nodes. The RS-configuration will consider not only availability of substrate's resources but also the interdependence between all {\it possible} mappings.
Specifically, this paper focuses on the following aspects: i) we propose to develop the theoretical model for constructing a RS-configuration for VNF mapping; ii) designing algorithms needed for mapping VNFs efficiently; and iii) conducting extensive experiments to validate the performance of the proposed model and algorithms.

\textbf{Organization:} We organize the paper as follows. In $\S$\ref{sec:background}, we introduce the background of RAN Slicing.
In $\S$\ref{sec:problem}, the network model as well as the research problem are presented.
Details of the proposed resource allocation algorithms are shown in $\S$\ref{sec:algorithm}. In $\S$\ref{sec:evaluation}, we discuss the experiment results. Finally, we make key concluding
remarks in $\S$\ref{sec:conclusion}.

\section{RAN Slicing and Current Challenges} \label{sec:background}

In order to focus on discussing the involved algorithmic problems, we present a generic architecture that is depicted in Fig. \ref{fig:ns}. The architecture in the figure shows the interactions of one (or more) {\it mobile network operator(s)} (MNOs) with multiple {\it enterprises}. The key aspect of RAN slicing is that the role of the MNO is to {\it coordinate} and allocate resources of the substrate network to ensure the harmonic coexistence of multiple RAN slices, while the role of the enterprise is to place slice requests and then manage the provided slices \cite{7926920,8004165,8057230,8004168}.
In particular, to prepare for a new slice initialization, the enterprise must first determine the required slice functionality and resources needed for VNFs of the requested slice.
It is envisioned that slice templates will be available for the most common types of services \cite{10.1145/2795381.2795390,7509393,7899415}. Thus, the enterprise may
select the slice template that fits its purpose and parameterize it according to its needs.
More specifically, a RAN slice that is independent from others consists of a set of VNFs. Upon receiving requests from the enterprise, the MNO works with the ``RS-configuration" (RAN
enforcement) to allocate  the substrate network resources (e.g., resource blocks in long-term evolution (LTE)) to VNFs and virtual links between VNFs, and provide an appropriate mapping plan, according to the proffered demand, RAN policy, connectivity, and available spectrum resources of the substrate nodes, also called {\it resource blocks} (RBs).

There have been efforts to design algorithms for the RAN resource allocation problem that tend to consider only the available resources of the substrate network to design a mapping plan and disregard the VNF connectivity and bandwidth requirements \cite{9071984,7529130,7417376,8377187}.
Again, a mapping plan for VNFs will not only hinge on substrate nodes resource allocation but also rely on the VNF connectivity requirement and bandwidth requirements of the VNF links \cite{7524565,7410276,8725795,7490359} (see $\S$\ref{sec:basic} for the details).
In order to be able to support novel services in a large-scale system, a more flexible and comprehensive RAN configuration scheme is needed to embed VNFs in the substrate nodes.

Unfortunately, existing works \cite{8717789,7946928,6449268,doro2019slice} consider constructing mapping plans for VNFs by considering only available resources, and this may cause an {\it internal fragmentation} in the available resources of the substrate network. In other words, available resources are sufficient but not {\it eligible} for mapping more VNFs onto the substrate nodes due to not meeting the bandwidth and connectivity requirements \cite{6963800,8494813}.
This step requires considerations of the interdependency property in the RAN slicing and efficient algorithms for mapping VNFs onto the substrate nodes. The algorithmic considerations of this step are described in the following sections of this paper. Based on the solution to this problem, the MNO can embed more VNFs in the substrate network using SDN and NFV technology \cite{8004168,8004169}.

\section{Network Model and Research Problem} \label{sec:problem}

In the following sections, the basic mathematical notations are initially introduced and simple examples are used to present the key ideas behind the proposed
RS-configuration paradigm and illustrate its advantages over the conventional resource allocation-based enforcement algorithm in $\S$\ref{sec:basic}. We then provide a general problem formulation and highlight the major objectives of the proposed paradigm in $\S$\ref{sec:general-problem}.

\subsection{Network Model: Basic Notations and Illustration} \label{sec:basic}
Given a substrate network $G^S=(N^S,E^S)$, let $N^S$ and $E^S$ denote the set of substrate nodes and links, respectively.
Considering a node $s \in N^S$, the total available resources at node $s$ is defined as $\mathcal{R}_s$. Namely, any node $s$ in $N^S$ can allocate a maximum amount $\mathcal{R}_s$ of resources to VNFs. A VNF in a specific slice is not available and accessible by other network slices for the isolation purpose \cite{5Gwhitepaper,8901876,7073808}.

Let $\mathcal{G}^V$ be the set of RAN slices running over the substrate network $G^S$, where $\mathcal{G}^V = \{G^{V_1}, G^{V_2}, \ldots, G^{V_{\ell}} \}$.
For emphasis, we will denote a single RAN slice as $G^{V_i}$ ($1 \leq i \leq \ell$) and drop $i$ when the context is clear.
We define a RAN slice $G^{V_i}=(N^{V_i},E^{V_i}) \in \mathcal{G}^V$, where $N^{V_i}$ and $E^{V_i}$ are the set of VNFs and the set of virtual links between VNFs, respectively.

For any VNF $u \in N^V$, it needs an amount $\Re_{u}$ of available resources at a substrate node to be embedded. For any virtual link $(u,v) \in E^V$, it requires a bandwidth $b_{(u,v)}$ for data flows between VNF $u$ and VNF $v$.
Likewise, instead of making the entire substrate network $G^S$ available for routing traffic of all data flows, we consider a more general case: we are restricted to the substrate network $G^S$ with an {\it available} capacity constraint matrix $\mathcal{C}$ for all substrate links $(s,t) \in E^S$, which specifies the maximum bandwidth $\mathcal{C}_{(s,t)}$ that can be allocated to all virtual links mapped onto $(s,t)$. We introduce a binary mapping variable $\mathcal{M}^u_s$ to indicate the decision that the VNF $u$ is mapped onto the substrate node $s$ when $\mathcal{M}^u_s = 1$, and $\mathcal{M}^u_s = 0$ otherwise.

In the first place of setting up the network, all VNFs must be {\it mapped} ({\it embedded}) onto substrate nodes.
To ensure that all VNFs will be embedded, the value of the binary mapping variable can be obtained as follows:
$\sum\limits_{s\in N^{S}}{\mathcal{M}_{s}^{u}}=1$ for all $u\in \bigcup\limits_{i=1}^{\ell }{N^{{{V}_{i}}}}$.
However, due to the limited capacity at substrate nodes (available resources) and links (maximum bandwidth), a VNF $u$ can actually be mapped onto a substrate node $s$ if the available resources at $s$ are sufficient ($\Re_u \leq \mathcal{R}_s$) and meet bandwidth and connectivity conditions such that $b_{(u,v)} \leq \mathcal{C}_{(s,t)}$, $(s,t) \in E^S$ for all $s,t \in N^S$, $v \in \mathcal{N}_u$, and $\phi^v_t = 1$, where $\mathcal{N}_u$ is the set of VNF $u$'s neighbors.
In addition, we use a variable $\psi _{(s,t)}^{(u,v)}$ to represent a simultaneous mapping of two VNFs $u$ and $v$ onto two substrate nodes $s$ and $t$, respectively. In particular, $\mathcal{M}_{s}^{u}+\mathcal{M}_{t}^{v}-\psi _{(s,t)}^{(u,v)}\le 1$\footnote{Such relations help avoid the quadratic constraint.}.

Fig. \ref{fig:ran1} shows an example of VNFs embedded in substrate nodes of the RAN slicing. The substrate network consists of five base stations (BSs). Each BS provides four resource blocks (RBs) (two frequency units during two-time slots) to RAN slices. In the figure, there are also two RAN slices, and each consists of various VNFs. Both slices utilize the same substrate  cellular network resources. We consider the VNF’s resource requirement, as each needs a certain amount of resources at a substrate node to be embedded, where $R_{u_1} = 2$ RBs, $R_{u_1} = R_{p_2} = R_{p_3}$ and $R_{p_2} = 1$ RB (25\% resources of a substrate node). In other words, if a VNF requires 25\% of the spectrum resources, the substrate node should make 25\% of the RBs  to the VNF. In the figure, the solid link between any two VNFs in a slice indicates that they are neighbors (e.g., $(u_1,u_2), (p_1,p_2), (p_1,p_3)$ and $(p_2,p_3)$). Such a relation as above implies that for any  VNF $u\in \bigcup\limits_{i=1}^{\ell }{N_{}^{{{V}_{i}}}}$, the target substrate nodes that the  VNFs can be embedded in must be the substrate nodes OR one of the neighbors of the substrate nodes on which the VNFs neighbors are embedded. This connectivity constrain make the VNF $u_1$ cannot be embedded in the substrate node $s_1$ even though there are sufficient available resources.

As illustrated by the simple example above, existing works \cite{8454739,8254073,8254685} that focus on mapping VNFs by considering only the available resource of substrate nodes may fail to fully embed all VNFs. Moreover, they cannot leverage the maximum number of VNFs embedded, resulting in malfunction in RAN slices. Therefore, it is vital to come up with efficient algorithms to handle
this resources allocation and embedding processes.
The resources allocation algorithms will allow us to attain higher flexibility and embedding performance simultaneously, without needing further resources. To illustrate how resources allocation are handled in this case, consider again the simple example in Fig. \ref{fig:ran2} with a mapping plan to embed VNFs following a strict order. The result is that all VNFs are successfully embedded in the substrate network. We remark that the embedding algorithm should guarantee that all mapped VNFs meet the bandwidth and connectivity requirements.

\begin{figure}[t]
\center
\subfigure[]{\includegraphics[width=4.2cm]{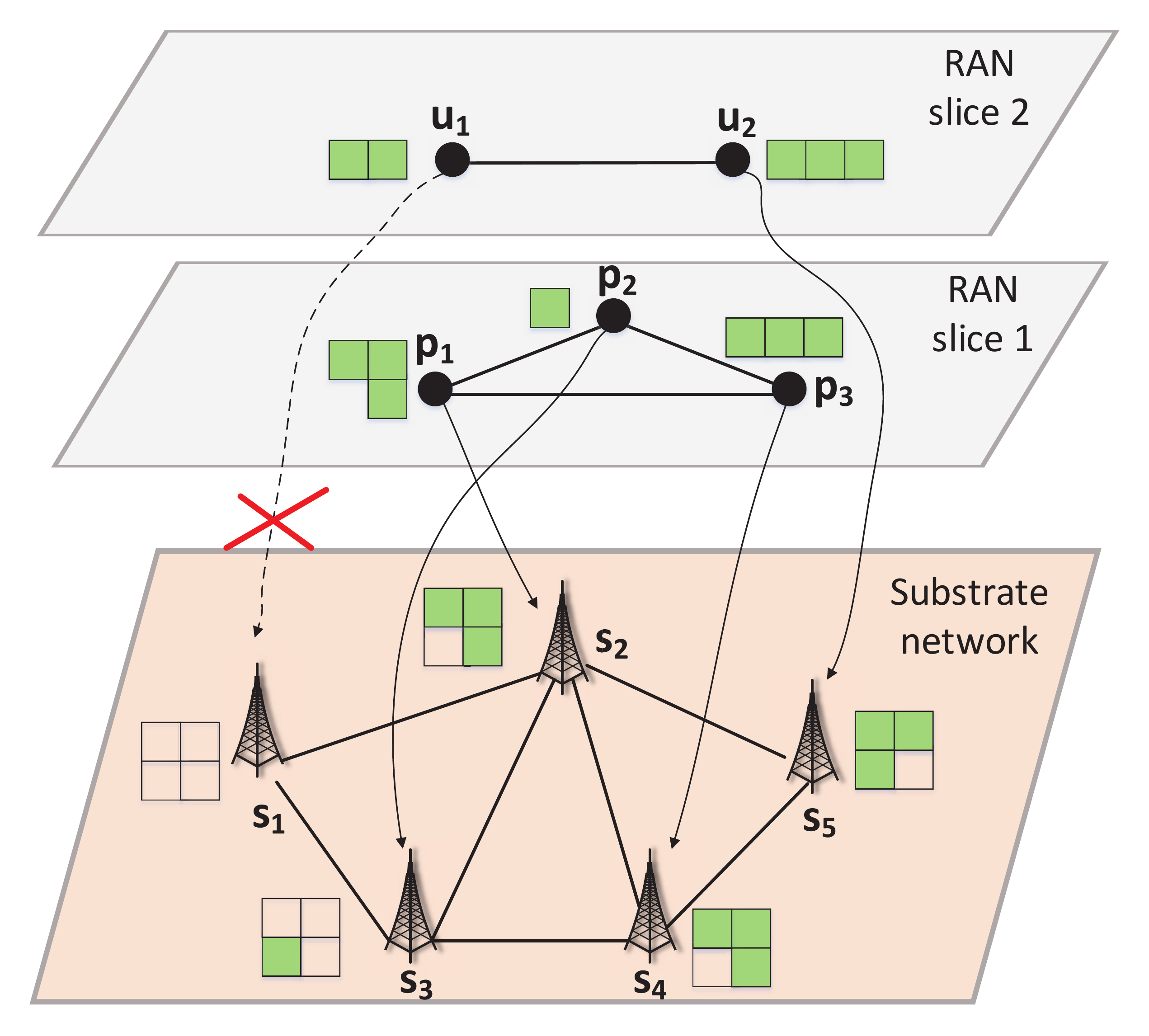} \label{fig:ran1}}
\subfigure[]{\includegraphics[width=4.2cm]{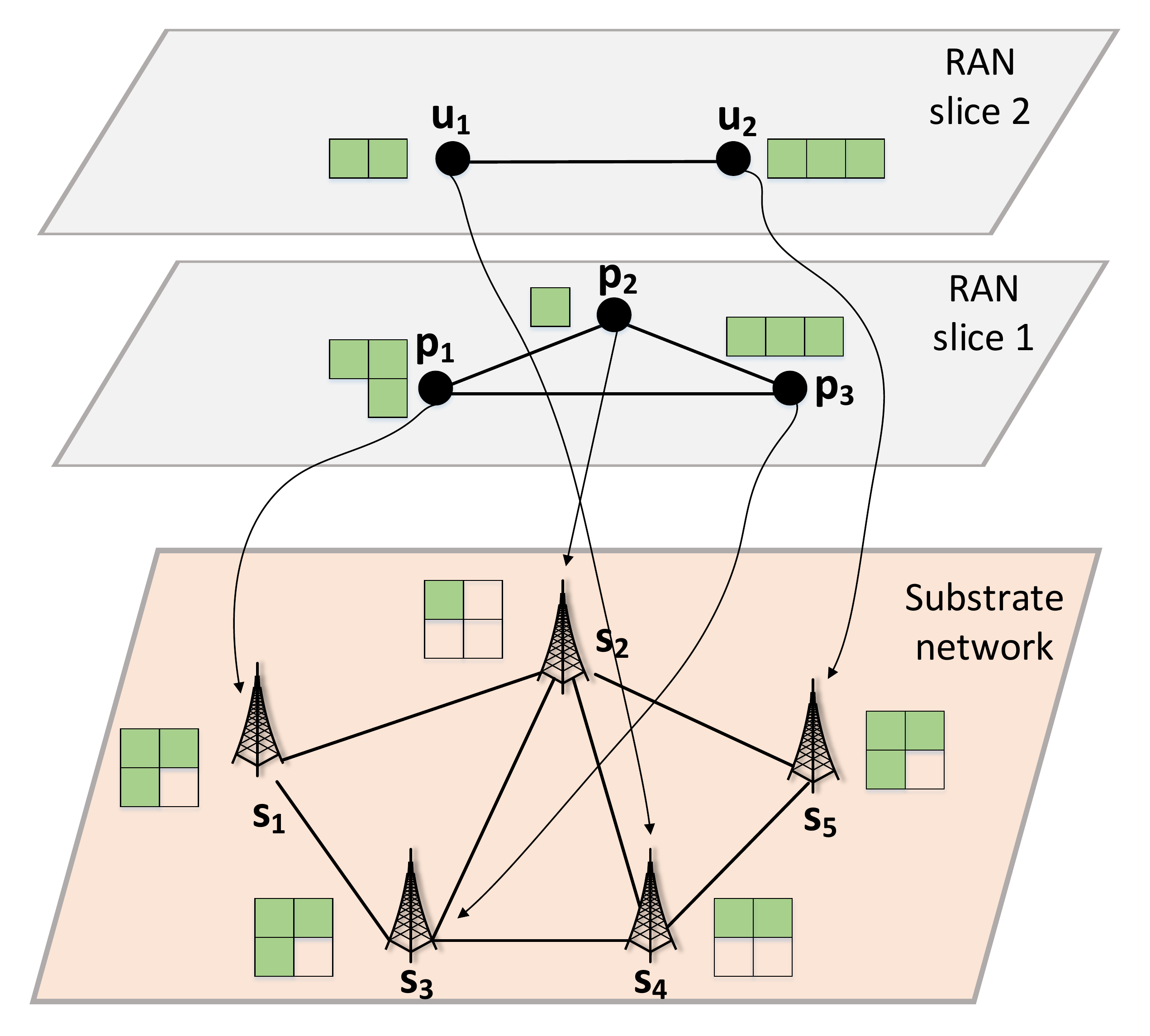} \label{fig:ran2}}
\caption{Resources allocation and embedding VNFs in a RAN slicing with a) a random mapping and b) RS-configuration algorithms (best viewed in color).}
\label{fig:mapping-slicing}
\vspace{-10pt}
\end{figure}

\subsection{RS-Configuration: Research Problem} \label{sec:general-problem}

Clearly embedding VNF strategy will not only hinge on substrate nodes resource allocation but also rely on the substrate connection set $E^S$ of the substrate network and the bandwidth requirement of the virtual links between VNFs. This is a key research challenge we will tackle in this paper.
In particular, for an RS-configuration in a cellular network, we would like to explicitly account for the {\it mapping plan} of VNFs, with the goal of providing a high performance for RAN slicing-based applications in terms of resources allocation for embedding VNFs in substrate node. A high performance in this stage is reflected through the number of successfully embedded VNFs.
This is in contrast to many existing studies where no recoverable and stable performance is assured under failures for RAN slicing-based applications, nor is any attempt made to mitigate the impact of the RAN configuration process on the entire network. We observe that, in general, the enhanced embedding performance is achieved at the expense of {\it increased configuration units} (more mappings) of the substrate network resources for all RAN slices.

\textbf{INSTANCE}: Suppose we are given a substrate network $G^S=(N^S,E^S)$ with a set of substrate nodes and links, respectively. Each substrate node $s \in N^S$ can allocate to VNFs an amount $\mathcal{R}_s$ of resources.
We consider a general case: we are restricted to the substrate network $G^S$ with an {\it available} capacity constraint matrix $\mathcal{C}$ for all $(s,t) \in E^S$, which specifies the maximum bandwidth $\mathcal{C}_{(s,t)}$ that can be allocated to all virtual links mapped onto $(s,t)$. Also, a set of RAN slices $\mathcal{G}^V = \{G^{V_1}, G^{V_2}, \ldots, G^{V_{\ell}} \}$ is running over the substrate network $G^S$. For any VNF, $u \in N^{V_i}$ requires an amount $\Re_{u}$ of available resources at a substrate node to be embedded. For any virtual link $(u,v) \in E^{V_i}$, it requires a bandwidth $b_{(u,v)}$ for data flows between two VNFs $u$ and $v$.

\textbf{QUESTION}: Does a mapping plan ($\mathcal{MP}$) for all VNFs exist in the RAN slicing, such that the number of successfully embedded VNFs is not less than $k$?
Mathematically, we formulate the following RS-configuration optimization problem:

\begin{equation}\label{eq:max-reccovery}
\underset{\mathcal{M}^u_s}{\mathop{\text{maximize}}}\,\sum\limits_{s\in N^{S}}{\sum\limits_{u\in \bigcup\limits_{i}^{\ell }{N^{{{V}_{i}}}}}{\mathcal{M}_{s}^{u}}}
\end{equation}

Here, the set of variables $\mathcal{MP} = \{\mathcal{M}^u_s\}$ represents one possible mapping/embedding plan solution for VNFs in the RAN slicing. The objective function is to maximize the total number of mappings of VNFs, which means that the number of VNFs embedded is maximized.
In addition, to ensure that one VNF can only be mapped onto at most one substrate node, they must meet the {\it mapping convergence constraint} at the first end $s$, $\sum\limits_{t}{\sum\limits_{v}{\psi _{(s,t)}^{(u,v)}}} = \mathcal{M}_{s}^{u}$, at the other end $t$, $\sum\limits_{s}{\sum\limits_{u}{\psi _{(s,t)}^{(u,v)}}} = \mathcal{M}_{t}^{v},$ and at both ends simultaneously, $\mathcal{M}_{s}^{u} + \mathcal{M}_{t}^{v}-\psi _{(s,t)}^{(u,v)}\le 1,$ for all $s,t\in N^{S}$, $u,v\in {{N}^{{{V}_{i}}}}$, and $\mathcal{M}_{s}^{u},$ $\mathcal{M}_{t}^{v},$ $\psi _{(s,t)}^{(u,v)} \in \{0,1\}$.

Apart from enhancing the embedding performance, we believe that latency is another key performance metric that must be accounted for in the design of the mapping plan for the RS-configuration. For example, delayed packets can cause the TCP timeout, resulting in unnecessary packet retransmissions, thus reducing overall application throughput.
We, therefore, desire to {\it bound the bandwidth variability}\footnote{Bounded bandwidth and bandwidth variability would help the end host at the network layer to reorder packets before passing to the transport layer.} in the mappings of VNFs as follows:
$\sum\limits_{u,v\in \mathcal{G}^V |\mathcal{M} _{s}^{u}=\mathcal{M} _{t}^{v}=1}{{{b}_{(u,v)}}}\le {{\mathcal{C}}_{(s,t)}}$ for all $(s,t) \in E^S$.
Namely, the total allocated bandwidth on a substrate link does not exceed its link capacity.
Likewise, the available resources at any substrate node $s$ must be sufficient, that is, $\sum\limits_{u \in \mathcal{G}^V | \mathcal{M}_{s}^{u}=1}{\Re_u\le \mathcal{R}_s}$ for all $s \in N^S$.
On the other hand, to conserve the virtual connection between VNFs, the mapping must also meet the connectivity constraint in which for any mapping of a VNF $u$ onto a substrate node $s$ and the other VNF $v$ onto the other substrate node $t$, for all $(u,v) \in E^V,$ $s$ and $t$ must be physically connected in the substrate network, that is, $(s,t) \in E^S$.
Thus, the mapping plan achieves the goal of providing a good embedding performance for RAN slicing.

\subsection{The Hardness of the RS-configuration Problem} \label{sec:difficulty}
In this section, we show the hardness of the RS-configuration problem by first showing that a sub-problem of the RS-configuration, referred to as the Reduced RS-configuration (RRS-configuration) problem, is NP-hard. The RRS-configuration is the RS-configuration with $\ell = 1$ that is there is only one network slice in the network, and $|N^S| = 1$ that is there is only one substrate node in the network. The problem is illustrated as follows:

\textbf{INSTANCE}: Given a substrate node $s$ that can allocate to VNFs an amount $\mathcal{R}_s$ of resources. Also, a RAN slice $\{G^{V_i}$ is running over the substrate node $s$. For any VNF, $u \in N^{V_i}$ requires an amount $\Re_{u}$ of available resources at a substrate node to be embedded.

\textbf{QUESTION}: Does a mapping plan ($\mathcal{MP}$) for all VNFs exist in the RAN slicing, such that the number of successfully embedded VNFs is not less than $k$?

Lemma \ref{lma:reduced-problem-hardness} is used to show that the
RRS-configuration problem is NP-hard. We first use the knapsack problem \cite{PISINGER20052271} to show the hardness of the RRS-configuration problem.
Then, the hardness of the RS-configuration problem is provided in Theorem \ref{thm:hardness}. The knapsack problem is illustrated
as follows:

\textbf{INSTANCE}: Given a set $n$ items; each item $i$ ($i = 1, \ldots, n$) has a value $c_i > 0$ and a weight $a_i > 0$, and a knapsack with capacity $b$.

\textbf{QUESTION}: Does the knapsack can be filled with items so as to maximize the total value of the items included in the knapsack?

\begin{lma}\label{lma:reduced-problem-hardness}
The RRS-configuration problem is NP-hard.
\end{lma}

\begin{proof}
In the knapsack problem, each item $i$ can be treated as a VNF
$u \in N^{V_i}$; the value $c_i$ can be considered as the value $\Re$ of each VNF; the knapsack with a capacity $b$ is treated as the substrate node $s$ with the capacity $\mathcal{R}_s$ in the network; and if we consider each item $i$ has a weight $a_i = 1$, the knapsack problem is also a RRS-configuration problem. We have that the knapsack problem is a sub-problem of the RRS-configuration problem. Because the knapsack problem is NP-hard \cite{PISINGER20052271}, the RRS-configuration is also NP-hard. This completes the proof.
\end{proof}

\begin{thm}\label{thm:hardness}
The RS-configuration problem is NP-complete.
\end{thm}

\begin{proof}
By Lemma \ref{lma:reduced-problem-hardness}, because the RRS-configuration problem, which is a sub-problem of the RS-configuration problem, is NP-hard, the RS-configuration problem is, therefore, NP-hard. In addition, it is clear that the RS-configuration problem belongs to the NP class. We therefore have that the RS-configuration problem is NP-complete.
\end{proof}

\section{Resource Allocation Algorithms} \label{sec:algorithm}

Equation \ref{eq:max-reccovery} lays out the RS-configuration optimization problem thereby giving a better understanding of the design of the resource allocation algorithms which achieve the given optimization. Since the RS-configuration optimization problem is NP-complete (shown in section \ref{sec:difficulty}), we propose two heuristic frameworks to solve the investigated problem: 1) we focus on leveraging a good embedding efficiency while maintaining the lowest possible time complexity, and 2) we consider the embedding efficiency as the vital requirement, and then try to maximize the number of embedded VNFs regardless the time complexity.
In the following sections, the first two algorithms (\ie the Resource-based Algorithm (RBA) and Connectivity-based Algorithm (CBA)) follow the first heuristic class given above. On the other hand, the Group-Connectivity-based Algorithm (GCBA) and Group-based Algorithm (GBA) seek to maximize the embedding performance without being subject to a specific time constraint function, thereby falling into the category of the second class of heuristic approaches. As part of each algorithm subsection, an overview and pseudocode are provided, along with a detailed description of the key idea behind the designed algorithms for the orderings and metrics involved.

\subsection{Resource-based Algorithm (RBA)} \label{subsec:rba}

Given a substrate network $G^S=(N^S,E^S)$, the major challenge becomes {\it how to select available resources and allocate them to suitable VNFs to achieve the highest number of embedded VNFs}. In this section, we propose the Resource-based Algorithm (RBA), which addresses the RS-configuration optimization problem (equation \ref{eq:max-reccovery}) using a heuristic approach. The main idea is to generate a mapping plan $\mathcal{MP}$ which provides good embedding efficiency while achieving low computational complexity. RBA achieves the heuristic by taking a greedy approach and prioritizing the embedding of VNF $u$ based on resources (from highest to lowest) $\Re_u$ $\forall\ u \in N^{V_i}$, to the substrate node with the highest available resources $\mathcal{R}_t\ |\ \mathcal{R}_t \geq \Re_u\ \forall\ t \in N^S$. The key idea behind this approach is to iteratively select the most demanding VNF (in terms of resources) and find the substrate node which can host it.


\begin{algorithm} [http]
\caption{Resource-Based Algorithm}
\label{alg:rba}
\textbf{Input:} $G^S, \mathcal{G}^V$ \\
\textbf{Output:} $\mathcal{MP}$
\begin{algorithmic}[1]

\State Initializing an empty set of mappings: $\mathcal{MP} = \emptyset$.

\State Constructing a descending-order set $\mathcal{L}^S$ of substrate nodes based on the available resource value $\mathcal{R}_{t'} \,\forall\, t' \in N^S$.

\State Constructing a descending-order set $\mathcal{L}^V$ of VNFs on the basis of Resources $\Re_{u'}\ \forall\ u' \in N^{V_i}$.

\While {$\mathcal{L}^V \neq \emptyset$}

    \State Let $u$ be an element in $\mathcal{L}^V$ with highest value $\Re_u$

    \State \textbf{EmbeddingVNF($u$, $\mathcal{L}^S$, $\mathcal{MP}$)}

    \State $\mathcal{L}^V = \mathcal{L}^V - \{u\} $

\EndWhile

\State return $\mathcal{MP}$

\end{algorithmic}
\end{algorithm}

Given the substrate network $G^S$ and a set of RAN slices $\mathcal{G}^V$, we follow the RBA with three main steps for generating the mapping plan $\mathcal{MP}$ as follows:

\begin{enumerate}
    \item We first construct the prioritized set of VNFs $\mathcal{L}^V$ based on the required resources (from highest to lowest).
    \item We then construct the descending-order set of substrate nodes $\mathcal{L}^S$ based on available resources.
    \item We finally generate the mapping plan $\mathcal{MP}$ for embedding all VNFs in $\mathcal{L}^V$ to the substrate nodes in $\mathcal{L}^S$.
\end{enumerate}

As indicated in the algorithm \ref{alg:rba}, we first initialize the $\mathcal{MP}$ set, the prioritized set of VNFs $\mathcal{L}^V$, and construct the descending-order set of substrate nodes $\mathcal{L}^S$ in the first three lines of the pseudocode.
The procedure EmbeddingVNF handles step 3 \ie the embedding process, and is described in detail from lines 3--\!16. For a VNF $u$ and a substrate node $t$ such that $u \in \mathcal{L}^V$ and $t \in \mathcal{L}^S$, the embedding process is separated into two cases, 1) VNF $u$ has no neighbor which is currently embedded onto a substrate node (lines 3--\!7), and 2) VNF $u$ has at least one neighbor which is currently embedded onto a substrate node (lines 8--\!14).

The motivation behind the separation of the embedding process into different cases is due to the different embedding constraints imposed upon a VNF in both cases. In the case one, the VNF $u$ is only subject to the resource constraint $\mathcal{R}_t \geq \Re_u$. This is because there are no existing neighbors embedded onto the substrate network for the VNF $u$ to ensure connectivity and bandwidth constraints for. Upon checking for the satisfaction of the resource constraint in line 4, the embedding process is completed and the mapping variable $\mathcal{M}^u_t$ is added to the mapping plan set $\mathcal{MP}$ (lines 5--\!6). In case two \ie when $u$ has at least one neighbor embedded onto a substrate node, the connectivity and bandwidth constraints need to also be satisfied to embed a VNF onto a substrate node. Firstly, the connectivity is achieved by constructing the descending-order set $\mathcal{L}^{\mathcal{N}_u}_{S^+}$, containing the substrate nodes onto which the neighbors of $u$ ($\mathcal{N}_u$), and the common neighbors of nodes in $\mathcal{N}_u$ are hosted (line 9). The representation of the set $\mathcal{L}^{\mathcal{N}_u}_{S^+}$ is expressed by the following mathematical relation:

\begin{equation} \label{eq:embedsubsset}
    \mathcal{L}^{\mathcal{N}_u}_{S^+} = \{ s_i,\ldots, s_j \} \bigcup \{\mathcal{N}_{s_i} \bigcap \ldots \bigcap \mathcal{N}_{s_j}\}\ |\  \mathcal{M}^{u_p}_{s_q} = 1\
\end{equation}

where $\mathcal{R}_{s_i} \geq \mathcal{R}_{s_j}$ for all $u_p \in \mathcal{N}_{u}$ and $s_q \in \{ s_i, \ldots, s_j \}$. Once the set $\mathcal{L}^{\mathcal{N}_u}_{S^+}$ has been constructed, we select the first substrate node $t_i$ in $\mathcal{L}^{\mathcal{N}_u}_{S^+}$, which satisfies the resource and bandwidth constraint to host the VNF $u$ (line 10). If the substrate node $t_i$ exists, VNF $u$ is embedded onto $t_i$ and $\mathcal{M}^{u}_{t_i}$ is added to $\mathcal{MP}$ (lines 11--\!13). Hence, VNFs within RAN slices are embedded onto the substrate network and the mapping plan $\mathcal{MP}$ is determined.

\subsection{Connectivity-based Algorithm (CBA)} \label{subsec:cba}


\begin{algorithm} \caption{Connectivity-Based Algorithm}\label{alg:cba}
\textbf{Input:} $G^S, \mathcal{G}^V$ \\
\textbf{Output:} $\mathcal{MP}$
\begin{algorithmic}[1]

\State Initializing an empty set of mappings: $\mathcal{MP} = \emptyset$.

\State Constructing a descending-order set $\mathcal{L}^S$ of substrate nodes based on the available resource value $\mathcal{R}_{t'} \,\forall\, t' \in N^S$.

\State Constructing a descending-order set $\mathcal{L}^V$ of VNFs on the basis of degree $|\mathcal{N}_{u'}|\ \forall\ u' \in N^{V_i}$.

\While {$\mathcal{L}^V \neq \emptyset$}

    \State Let $u$ be an element in $\mathcal{L}^V$ with highest value $|\mathcal{N}_u|$

    \State \textbf{EmbeddingVNF($u$, $\mathcal{L}^S$, $\mathcal{MP}$)}

    \State $\mathcal{L}^V = \mathcal{L}^V - \{u\} $

\EndWhile

\State return $\mathcal{MP}$

\end{algorithmic}
\end{algorithm}

While the set $\mathcal{MP}$ is obtained by RBA, we consider a second algorithm, Connectivity-based Algorithm (CBA) that can be done with the same computational complexity as well as can provide a better embedding performance. The essential idea behind CBA is to consider the set of VNF $\mathcal{L}^V$ based on the degree value $|\mathcal{N}_v|$ of all VNF $v \in N^{V_i}$. This allows the VNF $v'$ with the highest degree value $|\mathcal{N}_{v'}|$ to be embedded onto the substrate node $s$ which has the highest available resources $\mathcal{R}_s$ $\forall s \in N^S$. The newly ordered set $\mathcal{L}^V$ would increase the possibility of VNFs in the set $\mathcal{N}_{v'}$ to be embedded onto the substrate node set $\mathcal{N}_{s^+}$, where $\mathcal{N}_{s^+}$ denotes the set of substrate node $s$ and its neighbors $\mathcal{N}_s$ such that $\mathcal{N}_{s^+} = s \bigcup \mathcal{N}_s$


\begin{algorithm} [t]
\begin{algorithmic}[1]\label{proc:embedvnf}
\Procedure{EmbeddingVNF} {$u$, $\mathcal{L}^S$ $\mathcal{MP}$}

\State Let $t$ be an element in $\mathcal{L}^S$ with highest value $\mathcal{R}_t$

\If {$\mathcal{M}^v_s = 0\ \forall\ v \in \mathcal{N}_u, s \in N^S$}

    \If {$\mathcal{R}_t \geq \Re_u $}
        \State $\mathcal{R}_t = \mathcal{R}_t - \Re_u$
        \State $\mathcal{MP} = \mathcal{MP} \bigcup \{\mathcal{M}^u_t\}$
    \EndIf

\Else

    \State Constructing $\mathcal{L}^{\mathcal{N}_u}_{S^+}$, representing the descending-order set of substrate nodes and their common neighbors, based on the available resource value, onto which the substrate nodes are hosting $v$ $\forall\ v \in \mathcal{N}_u$

    \State Let $t_i$ be the first element in $\mathcal{L}^{N_u}_{S^+}\ |\ b_{(u, u')}  \leq \mathcal{C}_{(t_i, t_j)}$ and $\Re_u \leq \mathcal{R}_{t_i}, \forall\ u' \in \mathcal{N}_u,\ t_j \in \mathcal{N}_{t^+_i}$, where $\mathcal{N}_{t^+_i}$ denotes the set of $t_i$'s neighbors and itself.

    \If {$\exists\ t_i$}
        \State $\mathcal{R}_{t_i} = \mathcal{R}_{t_i} - \Re_u$
        \State $\mathcal{MP} = \mathcal{MP} \bigcup \{\mathcal{M}^u_{t_i}\}$
    \EndIf

\EndIf
\State return $\mathcal{MP}$

\EndProcedure
\end{algorithmic}
\end{algorithm}

In algorithm \ref{alg:cba}, to obtain the mapping plan $\mathcal{MP}$ the process is still hinged on the three main steps we presented in the previous section.
However, we try to improve the performance of the matching of the VNFs and substrate nodes by using a different metric. Specifically, we construct the descending-order set $\mathcal{L}^V$ based on the degree value $|\mathcal{N}_v|$ of every VNF $v \in N^{V_i}$ rather than the resource value $\Re_v$ (line 3 in the pseudocode of the algorithm \ref{alg:cba}). The following steps are similar to the algorithm \ref{alg:rba}, the construction of the descending-order set of substrate nodes $\mathcal{L}^S$ is executed in line 2 while the procedure EmbeddingVNF determines the mapping plan set $\mathcal{MP}$.

\subsection{Group-Connectivity-based Algorithm (GCBA)} \label{subsec:gcba}

In the later section ($\S$\ref{sec:evaluation}) we test the performance of all proposed algorithms and it will be verified that the algorithm \ref{alg:rba} and algorithm \ref{alg:cba} provide a good embedding efficiency within a fair time complexity. However, the RS-configuration's objective function is not completely optimized due to the heuristic for the RBA and CBA being subject to the constraint of a low computational complexity.
In this section we propose a high performance algorithm, referred to as Group-Connectivity-based Algorithm (GCBA) to attain the highest number of embedded VNFs, without being subject to the computational complexity constraint. GCBA achieves a good performance by considering VNFs in a specific list of clusters, namely $\mathcal{L}^V$. We then find the best matching substrate group for a cluster set $\mathcal{N}_{u_i^+}$ (refer to $\S$\ref{subsec:cba}) $\forall\ \mathcal{N}_{u_i^+} \in \mathcal{L}^V$

\begin{algorithm} \caption{Group-Connectivity-Based Algorithm}\label{alg:gcba}
\textbf{Input:} $G^S, \mathcal{G}^V$ \\
\textbf{Output:} $\mathcal{MP}$
\begin{algorithmic}[1]

\State Initializing an empty set of mappings: $\mathcal{MP} = \emptyset$.

\State Constructing a descending-order (based on the size of each set) list of clusters $\mathcal{L}^V$, where $\mathcal{L}^V = \{\mathcal{N}_{u_i^+}, \ldots, \mathcal{N}_{u_j^+} \}$ such that $\mathcal{N}_{u_i^+} \bigcap \mathcal{N}_{u_j^+} = \emptyset$.

\While {$\mathcal{L}^V \neq \emptyset$}

    \State Let $\mathcal{N}_{u_i^+}$ be the first element in $\mathcal{L}^V$

    \For {$ u \in \mathcal{N}_{u_i^+}$}

        \State \textbf{EmbeddingGroup$(u, \mathcal{MP})$}

    \EndFor

    \State $\mathcal{L}^V = \mathcal{L}^V - \{\mathcal{N}_{u_i^+}\}$

\EndWhile

\State return $\mathcal{MP}$;

\end{algorithmic}
\end{algorithm}

To embed the VNFs in a given cluster set ($\mathcal{N}_{u_i^+}$) onto the appropriate substrate node, GCBA  performs the three step $\mathcal{MP}$ set generation process as follows:

\begin{enumerate}
    \item Firstly, we construct the descending-order list of VNF cluster sets $\mathcal{L}^V$ based on the size of each cluster set.
    \item We then construct the ordered-set of substrate nodes $\mathcal{L}^S$ per VNF $u$ such that $\mathcal{R}_t \geq \Re_u\ \forall\ t \in \mathcal{L}^S$ and $u \in \mathcal{L}^V$.
    \item Finally, we embed the mutually exclusive cluster sets in $\mathcal{L}^V$ to the substrate network.
\end{enumerate}

As specified in algorithm \ref{alg:gcba}, we construct the list of VNF cluster sets in line 2. This process is carried out by the iterative addition of the VNF cluster head $v$ followed by its respective neighbor set $\mathcal{N}_v$. The representation of the set $\mathcal{L}^V$ is expressed by the following mathematical relation:

\begin{equation}\label{eq:clusterset}
    \mathcal{L}^V = [ \{v_i \bigcup v'_i\}, \{v_j \bigcup v'_j\}, \ldots , \{v_k \bigcup v'_k\} ]\
\end{equation}

where $|\mathcal{N}_{v_i}| \geq |\mathcal{N}_{v_j}| \geq |\mathcal{N}_{v_k}|\ \forall\ v'_i \in \mathcal{N}_{v_i}, v'_j \in \mathcal{N}_{v_j},$ and $v'_k \in \mathcal{N}_{v+k}$
We construct the set of substrate nodes $\mathcal{L}^S$, containing the substrate nodes which satisfy the resource constraint $\mathcal{R}_t \geq \Re_u\ \forall\ t \in N^S$ and $u \in \mathcal{L}^V$.
To optimize the embedding performance, we define the \textit{neighborhood resource cumulative property} $x(T)$ for all VNFs and substrate nodes in order to improve the selection process of the best possible substrate node. Given a VNF $v$ and a substrate node $s\ \forall\ v \in \mathcal{L}^V$ and $ s \in \mathcal{L}^S$, the $x(T)$ property for $v$ and $s$ is computed as follows:

\begin{equation}\label{eq:neighbrescumul}
\begin{gathered}
    v(T_\Re) = \Re_v + \sum\limits_{v' \in \mathcal{N}_v}{\Re_{v'}},\\
    s(T_\mathcal{R}) = \mathcal{R}_s + \sum\limits_{s' \in \mathcal{N}_s}{\mathcal{R}_{s'}}
\end{gathered}
\end{equation}

In equation \ref{eq:neighbrescumul}, the replacement of $x$ in $x(T)$ with $v$ and $s$ along with addition of $\Re$ and $\mathcal{R}$ as subscripts to $T$, allows for greater clarity while determining which element the $x(T)$ property is being computed for. In procedure EmbeddingGroup, we evaluate the $x(T)$ property (lines 3--\!4) after the embedding process is separated into two cases. For each case, the substrate node is picked based on $x(T)$.

Specifically, we first compute the difference $s(T_\mathcal{R}) - v(T_\Re)$ (lines 6 and 11), after that the substrate node $s$, with the positive difference value closest to $0$ is selected. In the case of only negative differences, the substrate node with the smallest negative difference is selected for embedding. The idea behind this metric (computing the difference) is to find the substrate node $s$ with the highest probability of supporting the VNF $u$ as well as the neighbor set $\mathcal{N}_u$. We select the substrate node with the difference value closest to $0$ to minimize the possibility of the available resources being wasted. In addition, prioritizing the positive difference ensures maximum embedding (based on the resource constraint) of all the VNFs $u'$ in $\mathcal{N}_{u^+}$, for all $\mathcal{N}_{u^+} \in \mathcal{L}^V$. Finally, if $\mathcal{L}^V$ consists of substrate nodes producing solely negative differences, the substrate node with the smallest difference is chosen so as to increase the probability of the maximum number of VNFs being embedded. The three steps to obtain the optimal $\mathcal{MP}$ set are therefore completed, thereby determining $\mathcal{MP}$.

\subsection{Group-based Algorithm (GBA)} \label{subsec:gba}

Similar to the algorithm \ref{alg:cba}'s improvement in the embedding efficiency as compared to the algorithm \ref{alg:rba}, can we achieve a greater optimization over the objective function in the RS-configuration (equation \ref{eq:max-reccovery}) in comparision to the algorithm \ref{alg:gcba}? We propose the Group-based Algorithm (GBA) to further enhance the optimization. To achieve the enhanced optimization, we chance the basis of ordering the set of VNFs $\mathcal{L}^V$. Instead of ordering the VNF cluster sets based on the basis of the degree of cluster heads, descending-order cluster sets are constructed, based on the the neighborhood resource cumulative value ($x(T)$) of the cluster heads. In doing so, the cluster sets with the greatest isolation of resources are embedded onto the substrate layer first.

\begin{algorithm} [t]
\caption{Group-Based Algorithm}\label{alg:gba}
\textbf{Input:} $G^S, \mathcal{G}^V$ \\
\textbf{Output:} $\mathcal{MP}$
\begin{algorithmic}[1]

\State Initializing an empty set of mappings: $\mathcal{MP} = \emptyset$.

\State Evaluate the value $v(T_{\Re})$ for all VNFs $v$ in $N^{V_i}$

\State Constructing a descending-order (based on the value of $v(T_\Re)$) list of clusters $\mathcal{L}^V$, where $\mathcal{L}^V = \{\mathcal{N}_{u_i^+}, \ldots, \mathcal{N}_{u_j^+} \}$ such that $\mathcal{N}_{u_i^+} \bigcap \mathcal{N}_{u_j^+} = \emptyset$.

\While {$\mathcal{L}^V \neq \emptyset$}

    \State Let $\mathcal{N}_{u_i^+}$ be the first element in $\mathcal{L}^V$

    \For {$ u \in \mathcal{N}_{u_i^+}$}

        \State \textbf{EmbeddingGroup$(u, \mathcal{MP})$}

    \EndFor

    \State $\mathcal{L}^V = \mathcal{L}^V - \{\mathcal{N}_{u_i^+}\}$

\EndWhile

\State return $\mathcal{MP}$;

\end{algorithmic}
\end{algorithm}

As shown in the algorithm \ref{alg:gba}, the step one (line 2) shows the improvement of GBA comparing to GCBA. Steps two and three, similar to the algorithm \ref{alg:gcba}, are carried out with the help of the procedure EmbeddingGroup.
The following steps are strictly followed the three steps of the embedding process we discussed in the earlier section. Hence, the $\mathcal{MP}$ set is obtained.


\begin{algorithm}
\begin{algorithmic}[1] \label{proc:embedgroup}
\Procedure{EmbeddingGroup}{$u$, $\mathcal{MP}$}

\State Constructing a set $\mathcal{L}^S$ of substrate nodes such that
$\mathcal{R}_{t'} \geq \Re_u \,\forall\, t' \in N^S$.

\State Evaluate the value $s(T_\mathcal{R})$ for all substrate nodes $s$ in $\mathcal{L}^S$
\State Evaluate the value $u(T_{\Re})$

\If {$\mathcal{M}^{v'}_{t'} = 0\ \forall\ v' \in \mathcal{N}_u, t' \in N^S$}
    \State Let $t$ be a substrate node in $\mathcal{L}^S$ with the minimum difference $s(T_\mathcal{R}) - u(T_\Re)\ \forall\ s \in \mathcal{L}^S$
    \State $\mathcal{R}_t = \mathcal{R}_t - \Re_u$
    \State $\mathcal{MP} = \mathcal{MP} \bigcup \{\mathcal{M}^u_t\}$
\Else

    \State Constructing $\mathcal{L}^{\mathcal{N}_u}_{S^+}$, representing the set of substrate nodes in $\mathcal{L}^S$ and their common neighbors, onto which the substrate nodes are hosting $v'$ $\forall\ v' \in \mathcal{N}_u$

    \State Let $t_i$ be a substrate node in $\mathcal{L}^{\mathcal{N}_u}_{S^+}$ with the minimum difference $s''(T_\mathcal{R}) - v(T_\Re)$ and $b_{(u, v')}  \leq \mathcal{C}_{(t_i, t_j)}\ \forall\ s'' \in \mathcal{L}^{\mathcal{N}_u}_{S^+}, v \in \mathcal{N}_{u^+},  v' \in \mathcal{N}_u,\ t_j \in \mathcal{N}_{t^+_i}$, where $\mathcal{N}_{t^+_i}$ denotes the set of $t_i$'s neighbors and itself.

    \If {$\exists\ t_i$}
        \State $\mathcal{R}_{t_i} = \mathcal{R}_{t_i} - \Re_u$
        \State $\mathcal{MP} = \mathcal{MP} \bigcup \{\mathcal{M}^u_{t_i}\}$
    \EndIf

\EndIf

\State return $\mathcal{MP}$;

\EndProcedure
\end{algorithmic}
\end{algorithm}

\section{Time Complexity} \label{sec:time}

\begin{thm}\label{thm:procedure}
The time complexity of the procedure EmbeddingVNF is bounded in $O(|N^{V_i}| \times |N^S|) + O(|E^{V_i}| \times |E^S|)$
\end{thm}

\begin{proof}
In the procedure EmbeddingVNF, the maximum number of VNF $u$'s neighbors is $|N^{V_i}| - 1$, so that it requires $O(|N^{V_i}|)$ to determine if the VNF $u$ has a neighbor already embedded onto a substrate node. If the VNF $u$ has at least one already embedded neighbor, the maximum number of substrate nodes we have to check for satisfaction of the connectivity constraint is $|N^S|$. Therefore, obtaining the substrate nodes in the set $\mathcal{L}^{\mathcal{N}_u}_{S+}$ requires $O(|N^{V_i}|^2 \times |N^S|)$, since we consider each of the neighbors of the VNF $u$ ($|N^{V_i}| - 1$), the links between the VNF $u$ and the neighbors of $u$ ($|N^{V_i}| - 1$), and the substrate nodes which satisfy the connectivity constraint for the VNF $u$ and it's neighbors ($|N^{S}|$). To achieve the descending-order sorting of the set $\mathcal{L}^{\mathcal{N}_u}_{S^+}$, we utilize the \textit{MergeSort} mechanism that requires $O(|N^S| \times \log{|N^S|})$ (best-worst case) to construct the set $\mathcal{L}^{\mathcal{N}_u}_{S^+}$. Likewise, it requires $O(|E^{V_i}| \times |E^S|)$ for checking the satisfaction of the bandwidth constraint. Finally, for determining the substrate node which satisfies the resource requirement, $O(|N^S|)$ is required. Hence the time complexity of the procedure is given as $O(|N^{V_i}|^2 \times |N^S|) + O(|E^{V_i}| \times |E^S|)$.
\end{proof}

\begin{thm}\label{thm:rba}
The time complexity of the RBA is bounded in $O(|N^{V_i}| \times ((|N^{V_i}| \times |N^S|) + O(|E^{V_i}| \times |E^S|)))$
\end{thm}

\begin{proof}
In the RBA, the descending-order sorting of the set of substrate nodes $\mathcal{L}^S$ requires $O(|N^S| \times \log{|N^S|})$. Similarly, to construct the descending-order set of VNFs $\mathcal{L}^V$, requires $O(|N^{V_i}| \times \log{|N^{V_i}|})$. Finally, the while loop iterates over the VNFs in $\mathcal{L}^V$, and for each iteration, the procedure EmbeddingVNF is executed. Considering that the number of VNFs in $\mathcal{L}^V$ is $|N^{V_i}|$, the while loop requires $O(|N^{V_i}| \times ((|N^{V_i}| \times |N^S|) + O(|E^{V_i}| \times |E^S|)))$ to determine the mapping plan. Hence, the time complexity of the RBA is given as $O(|N^{V_i}| \times ((|N^{V_i}|^2 \times |N^S|) + (|E^{V_i}| \times |E^S|)))$.

\end{proof}

\begin{thm}\label{thm:cba}
The time complexity of the CBA is bounded in $O(|N^{V_i}| \times ((|N^{V_i}| \times |N^S|) + O(|E^{V_i}| \times |E^S|)))$.
\end{thm}

\begin{proof}
In comparison to the theorem \ref{thm:rba}, the distinct preliminary step within the CBA is the metric, based on which the set of VNFs $\mathcal{L}^V$ are sorted. Regardless of the different metric, $O(|N^{V_i}| \times \log{|N^{V_i}|})$ is still required to construct the set $\mathcal{L}^V$. The set of substrate nodes $\mathcal{L}^S$ also requires $O(|N^S| \times \log{|N^S|})$ to be constructed. Finally, the while loop (similar to the theorem \ref{thm:rba}) requires $O(|N^{V_i}| \times ((|N^{V_i}| \times |N^S|) + O(|E^{V_i}| \times |E^S|)))$. Therefore, the time complexity of the CBA is given as $O(|N^{V_i}| \times ((|N^{V_i}|^2 \times |N^S|) + (|E^{V_i}| \times |E^S|)))$.
\end{proof}

\begin{thm}\label{thm:proceduregroup}
The time complexity of the procedure EmbeddingGroup is bounded in $O(|N^{V_i}|^2 \times |N^S|) + O(|E^{V_i}| \times |E^S|)$.
\end{thm}

\begin{proof}
In the procedure EmbeddingGroup, we first construct the set of substrate nodes $\mathcal{L}^S$ which satisfy the resource constraint, requiring $O(|N^S|)$. To evaluate the $s(T_{\mathcal{R}})$ property, we require to traverse through a maximum of $|N^S|$ substrate nodes, and for each substrate node the cumulative of $|N^S|$  substrate nodes (at maximum) requires to be computed. Therefore, we require $O(|N^S|^2)$ to compute the $s(T_{\mathcal{R}})$ property for each substrate node in $\mathcal{L}^S$. Likewise, to evaluate the $u(T_{\Re})$ property for the VNF $u$, we require $O(N^{V_i})$. As shown in theorem \ref{thm:procedure}, we require $O(|N^{V_i}|$ to determine whether the VNF $u$ has a neighbor already embedded. Computing the differences of the $x(T)$ property's of the substrate nodes in $\mathcal{L}^S$ and the VNF $u$ along with finding the substrate node $t$ with the minimum difference would require $O(|N^S|)$. To construct the set $\mathcal{L}^{\mathcal{N}_u}_{S+}$, $O(|N^{V_i}|^2 * |N^S|)$ would be required (similar to \ref{thm:procedure}). Lastly, we consider the substrate node $t_i$ in $\mathcal{L}^{\mathcal{N}_u}_{S+}$, with the minimum difference of the $x(T)$ property's of the substrate nodes (in $\mathcal{L}^{\mathcal{N}_u}_{S+}$)  and the VNF $u$, and satisfies the bandwidth constraint. To obtain the substrate node $t_i$, we require $O(|E^{V_i}| \times |E^S|) + O(|N^S|)$. Hence the time complexity of the procedure is given as $O(|N^{V_i}|^2 \times |N^S|) + O(|E^{V_i}| \times |E^S|)$.
\end{proof}

\begin{thm}\label{thm:gcba}
The time complexity of the GCBA is bounded in $O(|N^{V_i}|^2 \times ((|N^{V_i}|^2 \times |N^S|) + (|E^{V_i}| \times |E^S|)))$.
\end{thm}

\begin{proof}
In the GCBA, to construct the set of VNF cluster sets $\mathcal{L}^V$, requires $O(|N^{V_i}| \times \log{|N^{V_i}|}) + O(|N^{V_i}|^2)$. The while loop iterates over the VNF cluster sets  in $\mathcal{L}^V$, and each iteration executes the for loop which iterates through the cluster sets $N_{u_i^+}$. The procedure EmbeddingGroup is called within each iteration of the for loop, thereby resulting in the requirement of $O(|N^{V_i}|^2 \times ((|N^{V_i}|^2 \times |N^S|) + (|E^{V_i}| \times |E^S|)))$. Therefore, the time complexity of the GCBA is given as $O(|N^{V_i}|^2 \times ((|N^{V_i}|^2 \times |N^S|) + (|E^{V_i}| \times |E^S|)))$.
\end{proof}

\begin{thm}\label{thm:gba}
The time complexity of the GBA is bounded in $O(|N^{V_i}|^2 \times ((|N^{V_i}|^2 \times |N^S|) + (|E^{V_i}| \times |E^S|)))$.
\end{thm}

\begin{proof}
In comparison to the theorem \ref{thm:gcba}, GBA performs the sorting of the VNF cluster sets based on the $x(T)$ property of the cluster heads. Even with the change in the metric of sorting, the requirement $O(|N^{V_i}| \times \log{|N^{V_i}|}) + O(|N^{V_i}|^2)$ remains. The nested loop structure in the GBA requires $O(|N^{V_i}|^2 \times ((|N^{V_i}|^2 \times |N^S|) + (|E^{V_i}| \times |E^S|)))$. Therefore, the time complexity of the GBA is given as $O(|N^{V_i}|^2 \times ((|N^{V_i}|^2 \times |N^S|) + (|E^{V_i}| \times |E^S|)))$.
\end{proof}

\section{Performance Evaluation} \label{sec:evaluation}
In this section, we present simulations to demonstrate our approaches' efficiency by firstly implementing the proposed algorithms for constructing a embedding plan for RS-configuration through simulations. It is expected that the advantage of the algorithms will be manifested when the number of successfully embedded VNFs is significant high.

We consider the RAN slicing consisting of one substrate network and a variety of network slices. The number of network slices and the number of VNFs on each slice are randomly generated depending on the test case (we will present the details of setting in the following sections).
The substrate network is composed of a random number of substrate nodes representing the number of data centers. Each is initialized with an amount of resources.
It is feasible to create different embedding scenarios in the simulations and validate the simulation results. The first case is to test the feasibility of the proposed algorithms under the normal network condition that can provide reasonable resources for embedding VNFs into the substrate network, referred to as the \textbf{normal case}. In this case, we generate the number of substrate nodes from the interval $[60, 100]$, each is initialized with an amount of resources selected from the interval $[4, 8]$.
In terms of the network slice, the number of network slices is randomly generated from the interval $[2, 10]$ and the number of VNFs on each slice is also randomly chosen based on the number of network slices from the interval $[10, 100]$.

In the second case, we generate the network topology with limited resources that can test the performance of the proposed algorithms under the resource shortage condition, referred to as the \textbf{shortage case}. In this case, network resources are strictly controlled as follows: for the number of substrate nodes, from the interval $[60, 100]$; for the number of resources for each substrate node, from the interval $[2, 4]$; for the number of network slices, from the interval $[2, 10]$; and for the number of VNFs on each slice, from the interval $[1, 10]$.
In the following sections, Fig. \ref{fig:a-embedding} and Fig. \ref{fig:a-resource} show the results in the ``normal case" test, and Fig. \ref{fig:r-embedding} and Fig. \ref{fig:r-resource} show the results in the ``shortage case" test.

\subsection{Embedding Performance}

\begin{figure}[http]
\center
\subfigure[]{\includegraphics[width=4.2cm]{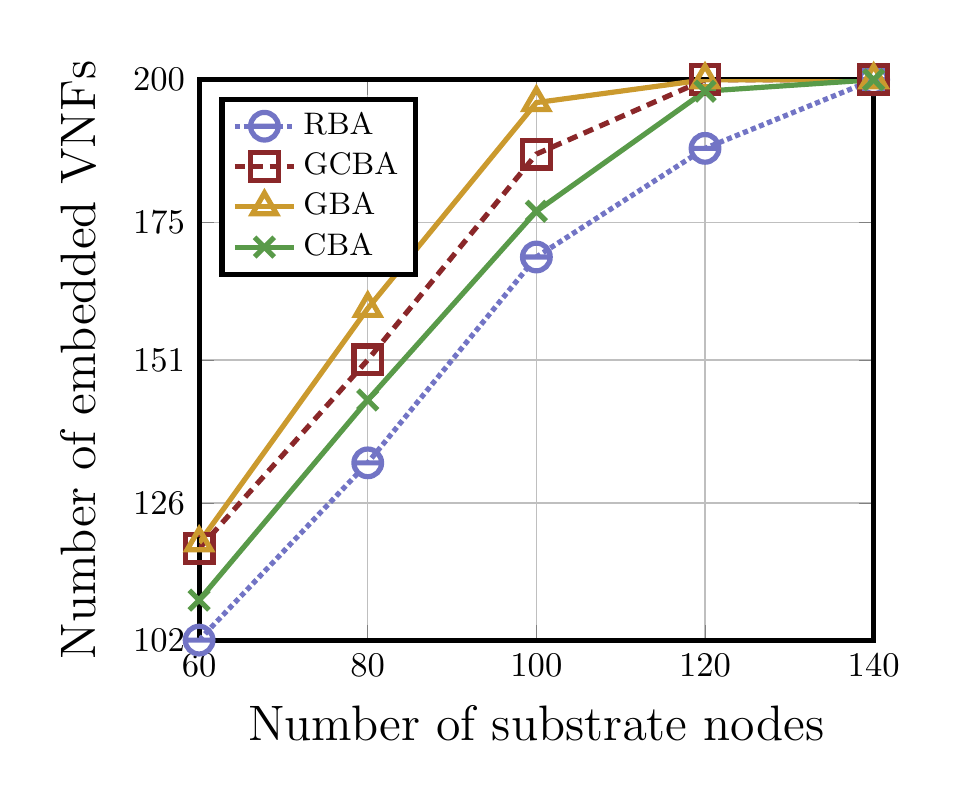}\label{fig:a-subnode-numsuc}}
\subfigure[]{\includegraphics[width=4.2cm]{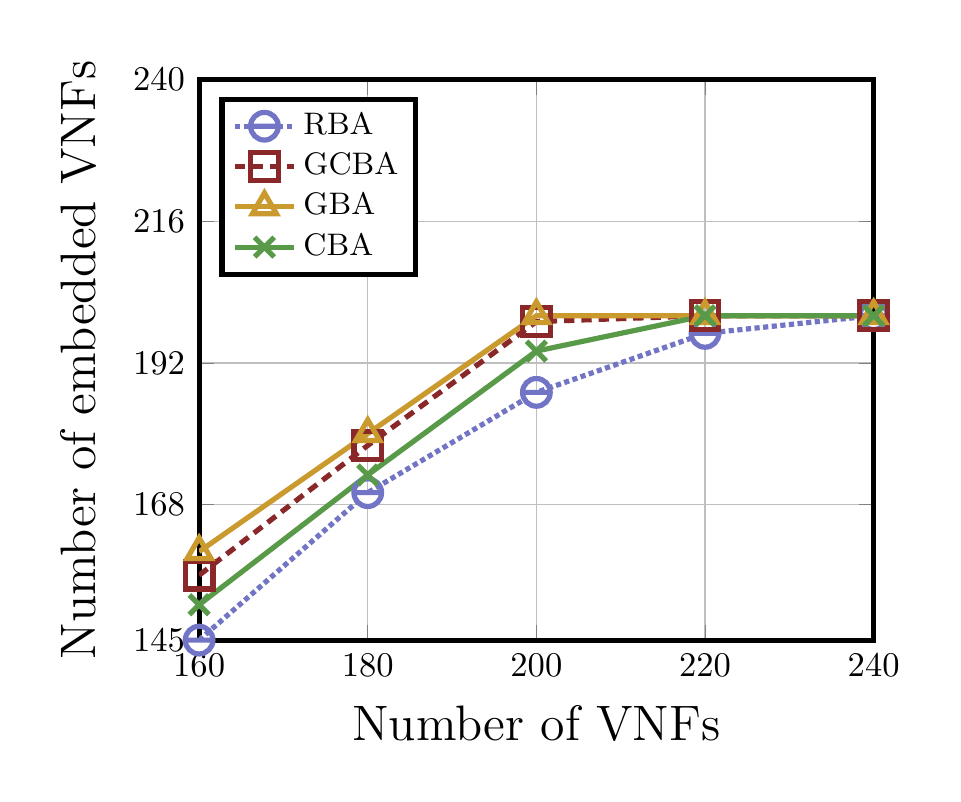} \label{fig:a-vnf-numsuc}}
\subfigure[]{\includegraphics[width=4.2cm]{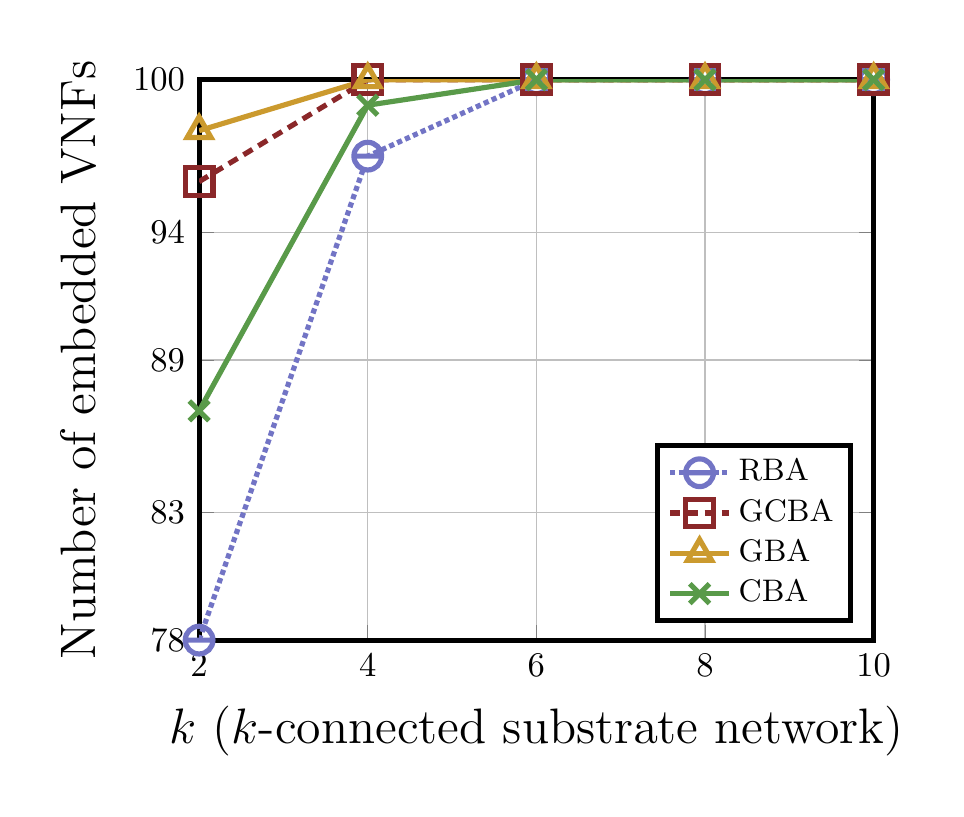} \label{fig:a-subdeg-numsuc}}
\subfigure[]{\includegraphics[width=4.2cm]{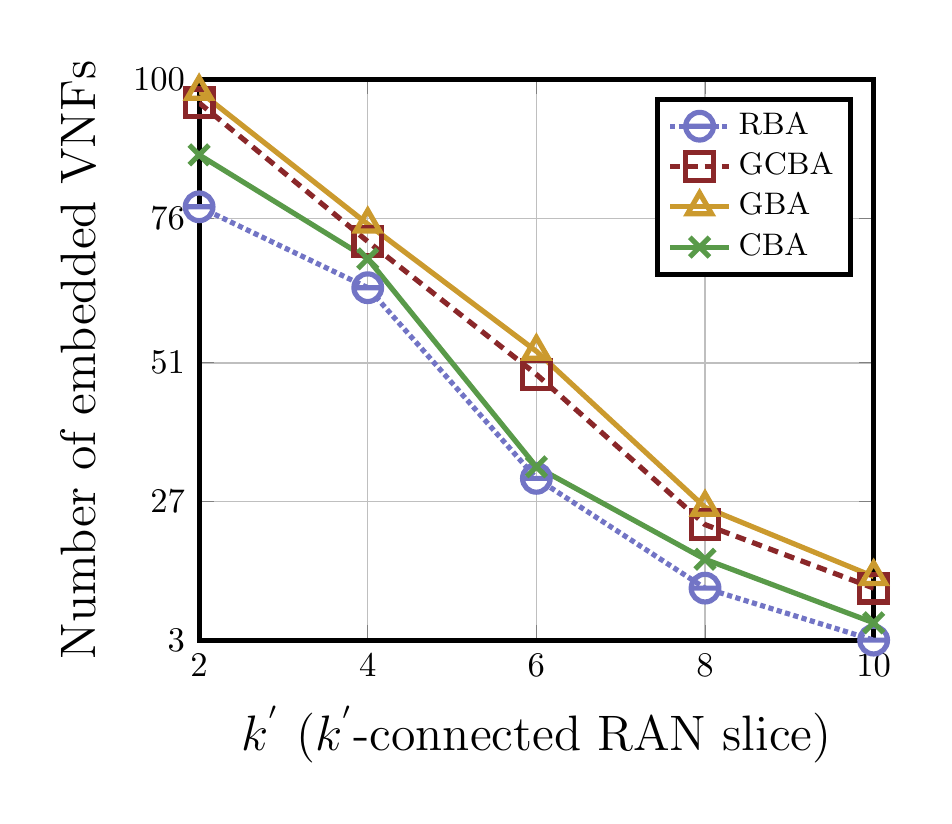}
\label{fig:a-degvnf-numsucc}}
\caption{\textbf{Normal Case:} total number of embedded VNFs when a) the number of substrate nodes ranging from 60 to 140, b) number of VNFs ranging from 160 to 240, c) $k$-connected substrate network with degree ($k$) ranging from 2 to 10 per substrate node, and d) $k'$-connected RAN slice with degree ($k'$) ranging from 2 to 10 per VNF.}
\label{fig:a-embedding}
\end{figure}

\begin{figure}[http]
\center
\subfigure[]{\includegraphics[width=4.2cm]{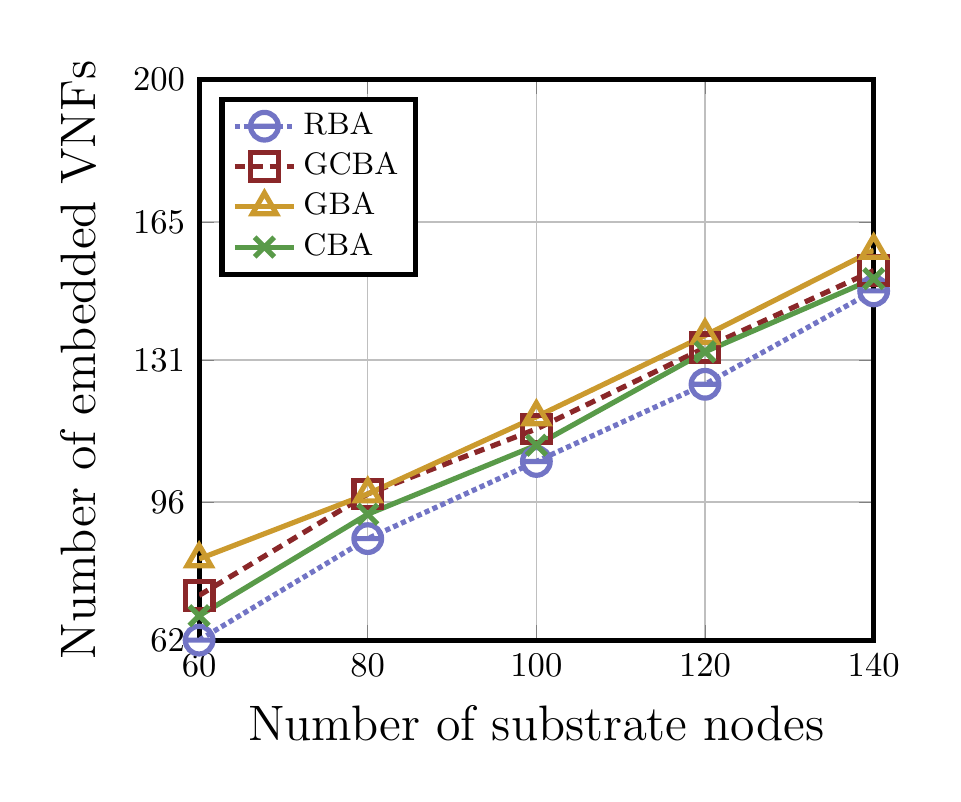} \label{fig:r-subnode-numsucc}}
\subfigure[]{\includegraphics[width=4.2cm]{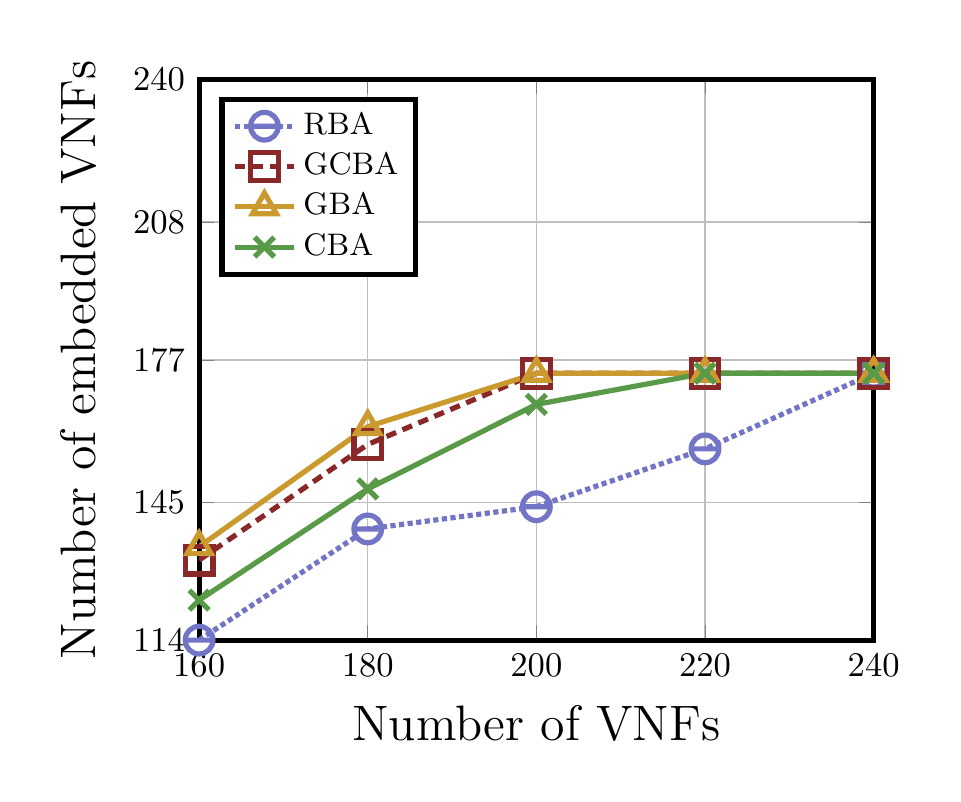} \label{fig:r-vnf-numsucc}}
\subfigure[]{\includegraphics[width=4.2cm]{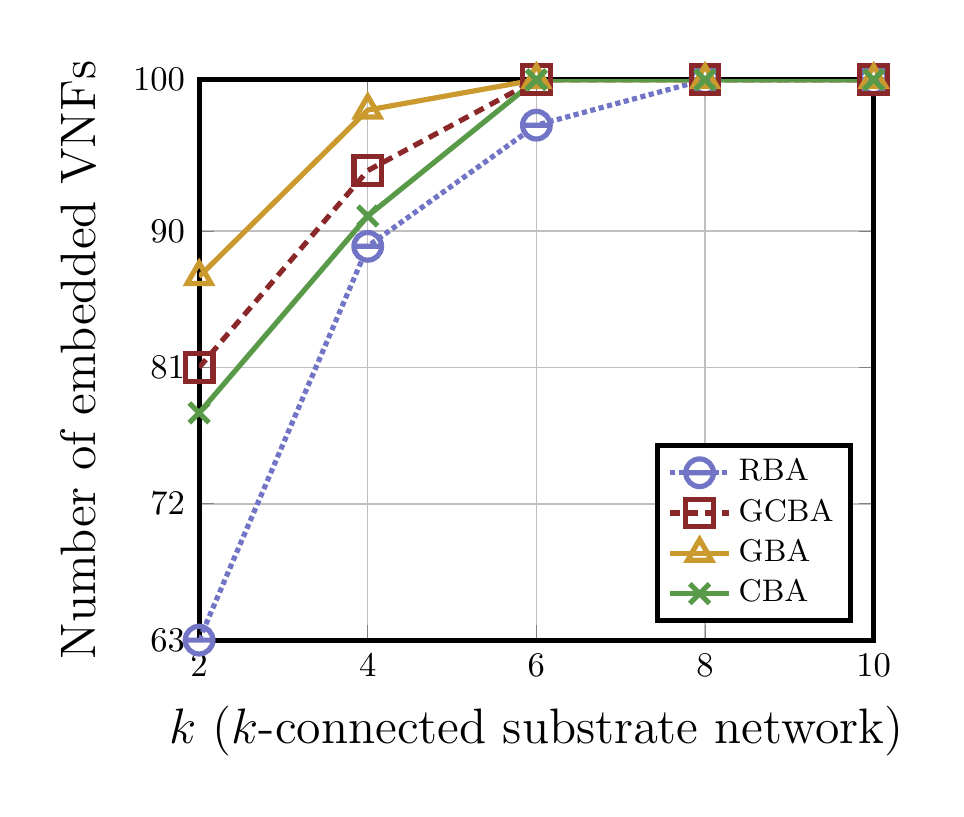} \label{fig:r-degsubnode-numsucc}}
\subfigure[]{\includegraphics[width=4.2cm]{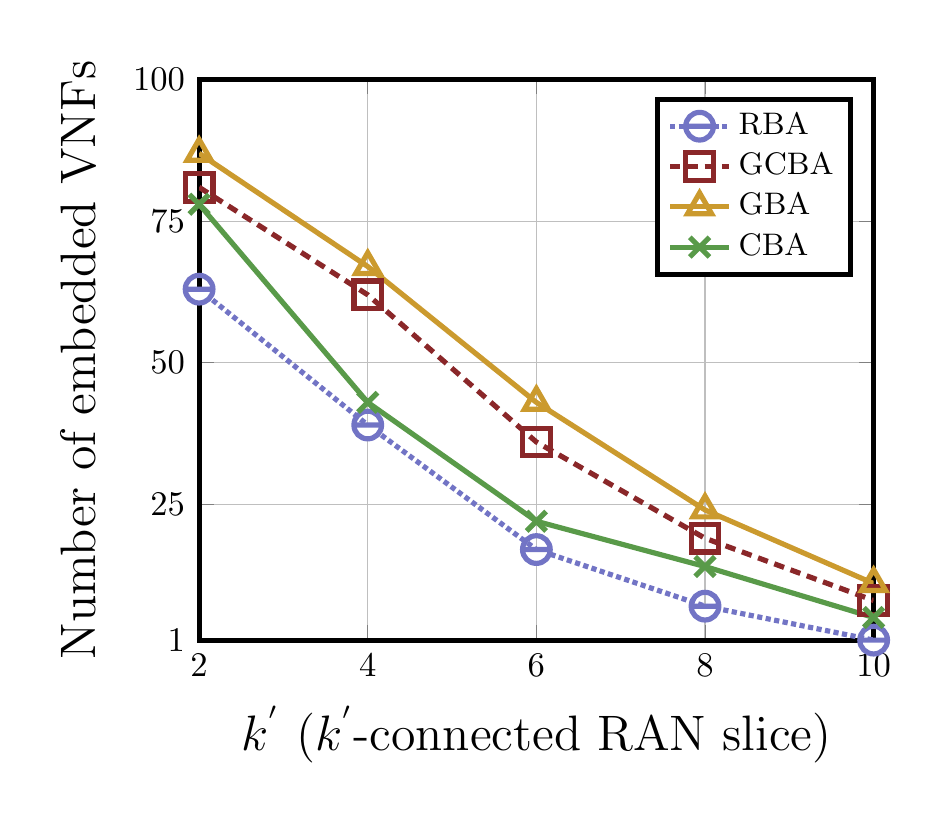} \label{fig:r-degvnf-numsucc}}
\caption{\textbf{Shortage Case:} total number of embedded VNFs when a) the number of substrate nodes ranging from 60 to 140, b) number of VNFs ranging from 160 to 240, c) $k$-connected substrate network with degree ($k$) ranging from 2 to 10 per substrate node, and d) $k'$-connected RAN slice with degree ($k'$) ranging from 2 to 10 per VNF.}
\label{fig:r-embedding}
\end{figure}

We first test the embedding performance of the proposed algorithms by comparing the results in both
normal and shortage case.
Fig. \ref{fig:a-embedding} and Fig. \ref{fig:r-embedding} show the performance of the proposed algorithms in terms of number of embedded VNFs in the normal and shortage case, respectively.

Fig. \ref{fig:a-subnode-numsuc} shows comparisons of the number of embedded VNFs when the number of substrate nodes ranges from $60$ to $140$. In the figure, GBA and GCBA provide the best performance with a higher total number of embedded VNFs.
During the embedding process, GBA and GCBA consider not only the required resources of a single VNF but also we take into account of the required resources of a cluster of VNFs. In addition, in the GCBA, we also consider the degree of VNFs (number of neighbors) when determining the mapping plan. The concept of cluster (interdependency) makes the GBA and GCBA to be the best in performing embedding VNFs.
The differences are reflected through the \textit{neighborhood resource cumulative property} $x(T)$ considered in both GBA and GCBA.

In Fig. \ref{fig:a-vnf-numsuc} we test the performance of the proposed algorithms when the number of VNFs ranging from $160$ to $240$. The results are similar to the Fig. \ref{fig:a-subnode-numsuc} as GBA and GCBA still demonstrate a better performance than the CBA and RBA in terms of the total number of embedded VNFs. We realize that the higher the VNFs in the substrate network, the higher the number of VNFs is successfully embedded into the substrate network. It indicates that the algorithms perform more efficiently when having more resources in the network (higher number of VNFs).

Fig. \ref{fig:a-subdeg-numsuc} shows comparisons of the test performance of the algorithms proposed for $k$, where $k$ denotes the degree of every substrate node in the substrate network, which ranges from $2$ to $10$. In the figure, GBA and GCBA demonstrate a better performance with a higher total number of embedded VNFs when comparing to CBA and RBA. We also observe that the higher the degree of substrate nodes, the higher the number of VNFs embedded in the substrate network. In addition, there is higher efficiency in performance when the degree of substrate nodes is increased.

In the Fig. \ref{fig:a-degvnf-numsucc} we analyze comparisons of the capability of the proposed algorithms with different $k^{'}$, in which $k^{'}$ denotes the degree of VNFs in RAN slices ranging from 2 to 10. GBA and GCBA continue to perform better when comparing CBA and RBA with the total embedded VNFs. We can evaluate that the lower the degree of VNFs in a RAN slice, the higher the number of VNFs successfully embedded in the substrate network. If the VNFs have a higher degree of VNFs, they demonstrate a lower embedding performance.
In summary, the results in Fig. \ref{fig:a-embedding} show the performance of all algorithms in the ``normal case" test.
We test the proposed algorithms using different metrics, including varying
the number of substrate nodes, changing the number of VNFs, adjusting the degree of substrate nodes ($k$) and the degree of VNF ($k'$).


In Fig. \ref{fig:r-embedding} we test the performance of the proposed algorithms in the ``shortage case" test to see how the shortage resource condition impacts to the number of VNF embedded in the network. Specifically, in Fig. \ref{fig:r-subnode-numsucc} we compare the embedding performance of the proposed algorithms to that the number of substrate nodes ranging from $60$ to $140$.
As shown in the figure, GBA and GCBA have the best efficiency performance with a higher total number of embedded VNFs. As mentioned in the previous analysis, both GBA and GCBA consider the \textit{neighborhood resource cumulative property} $x(T)$. In addition, in GBA we evaluate the $x(T)$ for both virtual networks and substrate network while
only the substrate network is evaluated with the $x(T)$ in GCBA. These metrics make GBA and GCBA perform better than CBA and RBA.
In contrast, using CBA and RBA, the embedding performance is lower with a lower number of embedded VNFs. Note that for CBA and RBA, all VNFs are scheduled to be embedded by considering only the embedding status of the individual VNFs resource or connectivity constraint.

Fig. \ref{fig:r-vnf-numsucc} shows the impact of the number of embedded VNFs with respect to the number of VNFs ranging from $160$ to $240$. The total number of VNFs embedded in the substrate network is higher in GBA and GCBA. Therefore, the efficiency of RBA and CBA is lower as the total number of embedded VNFs is lower in the substrate network. As the algorithms demonstrate a greater efficiency performance when number of VNFs is increased.
In the Fig. \ref{fig:r-degsubnode-numsucc} we analyze comparisons of the algorithms of number of embedded VNFs capability to $k$, where $k$ denotes the degree of each substrate node, which ranges from $2$ to $10$. GBA and GCBA exhibits a greater efficiency performance with a higher total number of VNFs embedded in the substrate network whereas, RBA and CBA  have a lower total number of VNFs embedded in the substrate network. We observe that the greater the degree of substrate nodes, the better the number of VNFs embedded in the substrate network. It indicates that there is a greater efficiency performance when the degree of substrate nodes is greater.
Fig. \ref{fig:r-degvnf-numsucc} shows comparisons of the total number of embedded VNFs to each ($k^{'}$). The number of VNFs connected per RAN slice is denoted by $k^{'}$, ranging from $2$ to $10$. GBA and GCBA show a higher total number of embedded VNFs in the substrate network. In contrast, RBA and CBA have a fewer total number of embedded VNFs as they have a lower embedding efficiency. In addition, the proposed algorithms demonstrate a lower embedding performance when the degree of VNF is higher.

\subsection{Resource Utilization Efficiency}

\begin{figure}[http]
\center
\subfigure[]{\includegraphics[width=4.2cm]{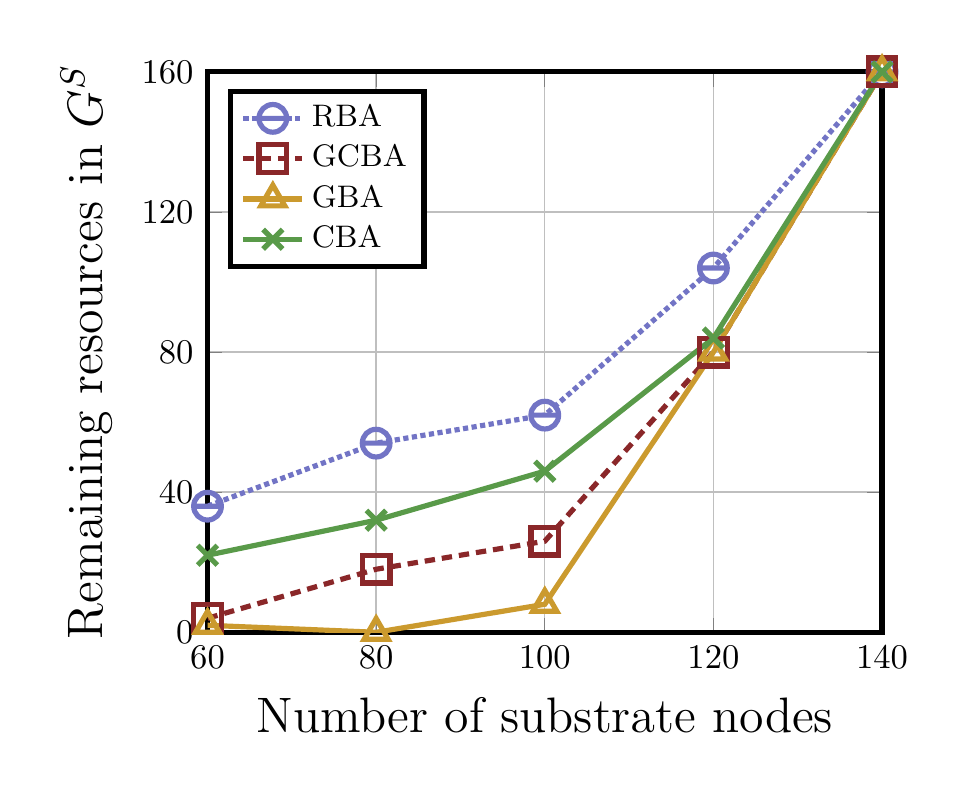} \label{fig:a-subnode-numrecava}}
\subfigure[]{\includegraphics[width=4.2cm]{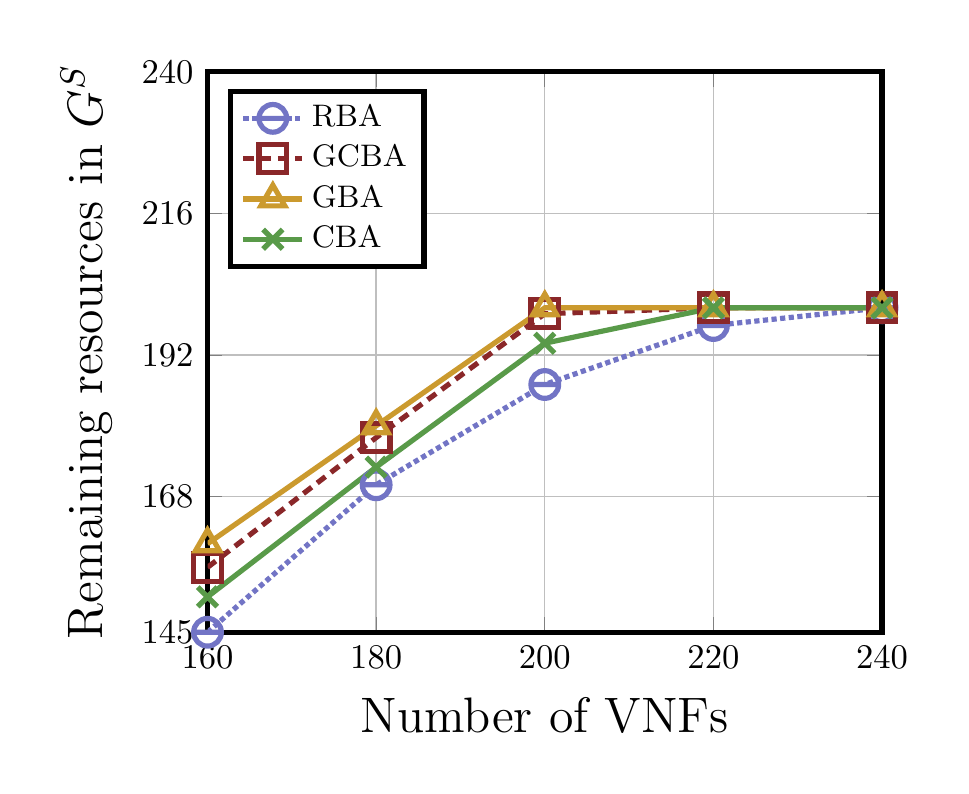} \label{fig:a-vnf-num-recava}}
\subfigure[]{\includegraphics[width=4.2cm]{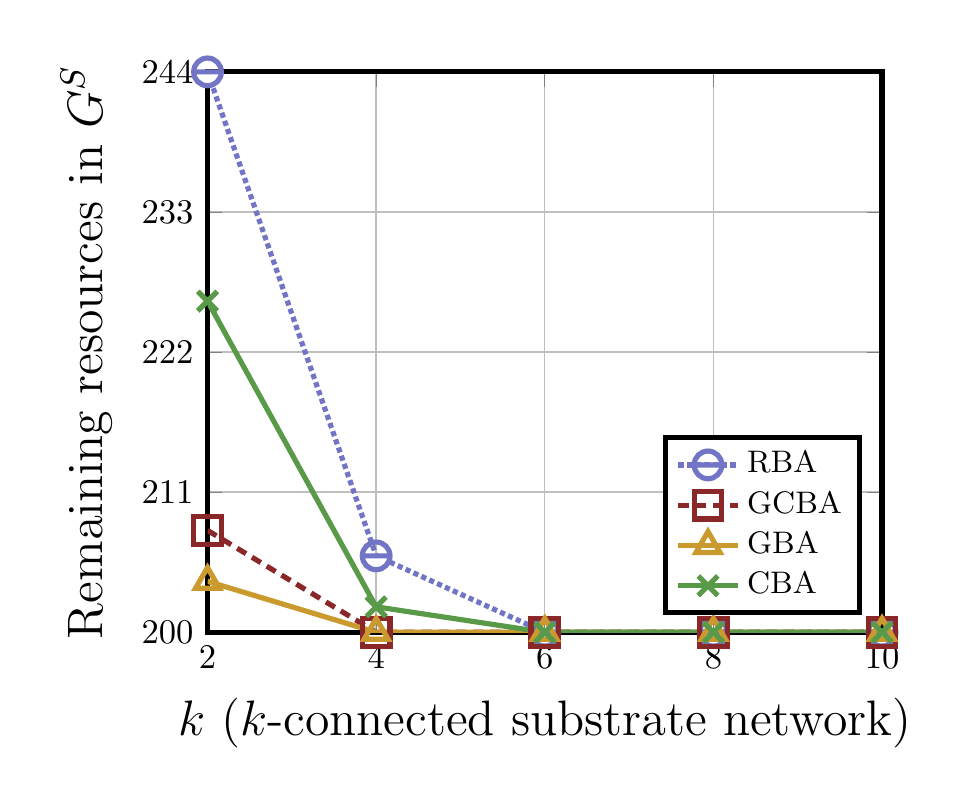} \label{fig:a-degsubnode-recava}}
\subfigure[]{\includegraphics[width=4.2cm]{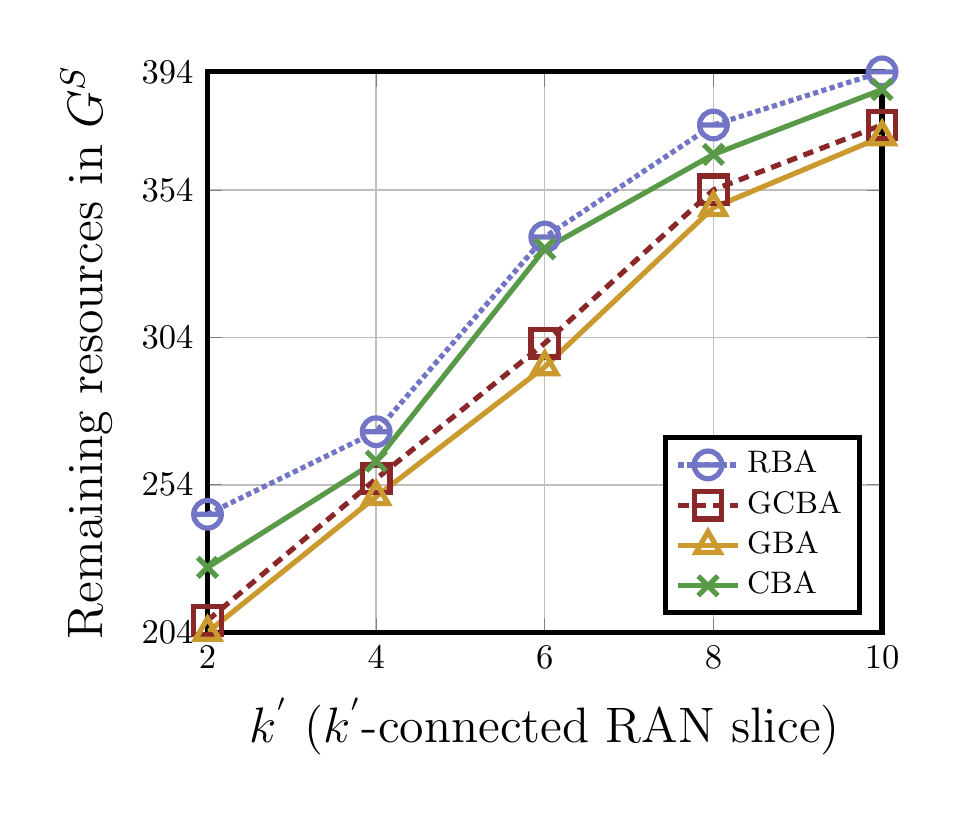} \label{fig:a-degvnf-recava}}
\caption{\textbf{Normal Case:} total amount of available resources in the substrate network when a) the number of substrate nodes ranging from 60 to 140, b) number of VNFs ranging from 160 to 240, c) $k$-connected substrate network with degree ($k$) ranging from 2 to 10 per substrate node, and d) $k^{'}$-connected RAN slice with degree ($k^{'}$) ranging from 2 to 10 per VNF.}
\label{fig:a-resource}
\end{figure}

\begin{figure}[http]
\center
\subfigure[]{\includegraphics[width=4.2cm]{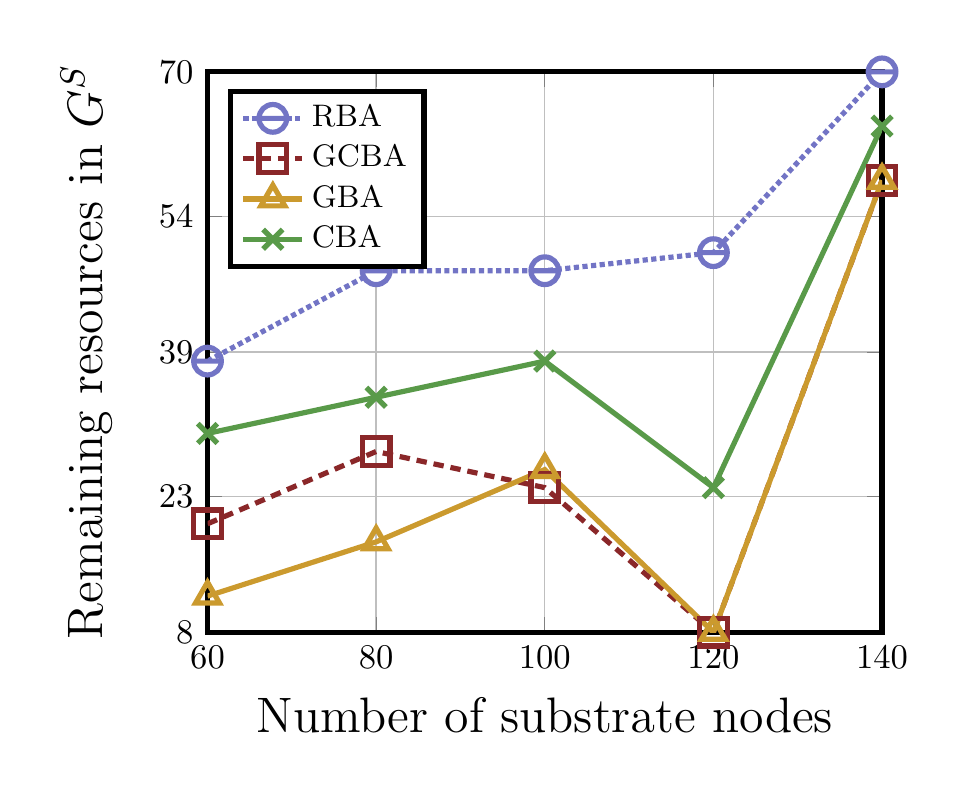} \label{fig:r-subnode-recava}}
\subfigure[]{\includegraphics[width=4.2cm]{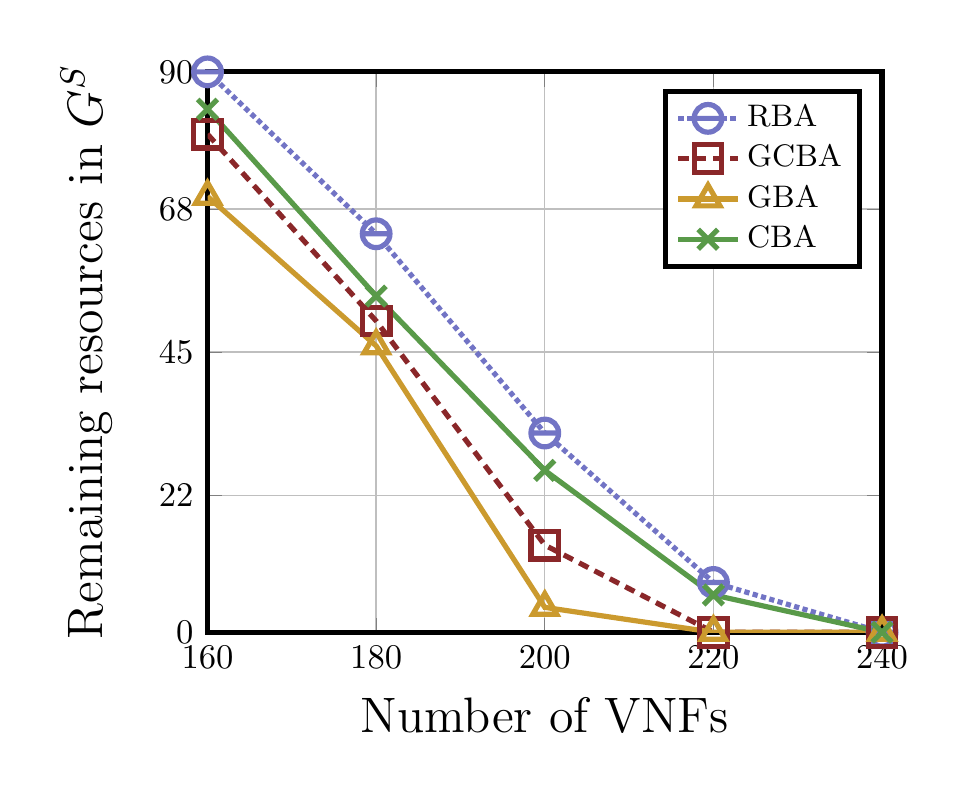} \label{fig:r-vnf-numrecava}}
\subfigure[]{\includegraphics[width=4.2cm]{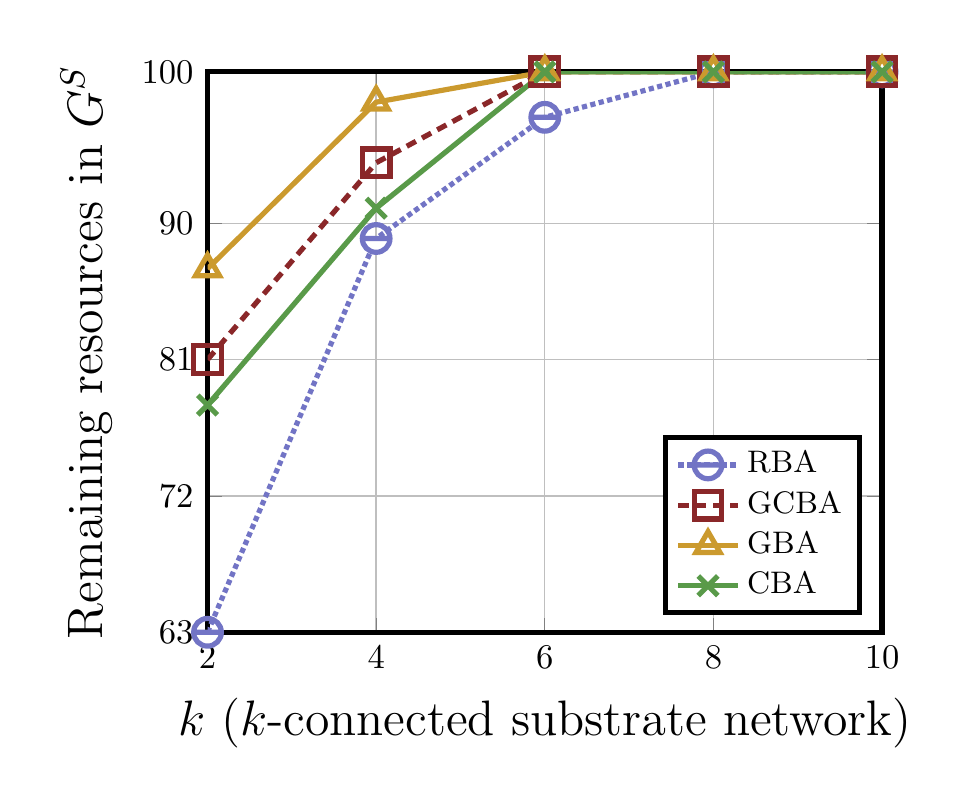} \label{fig:r-degsubnode-recava}}
\subfigure[]{\includegraphics[width=4.2cm]{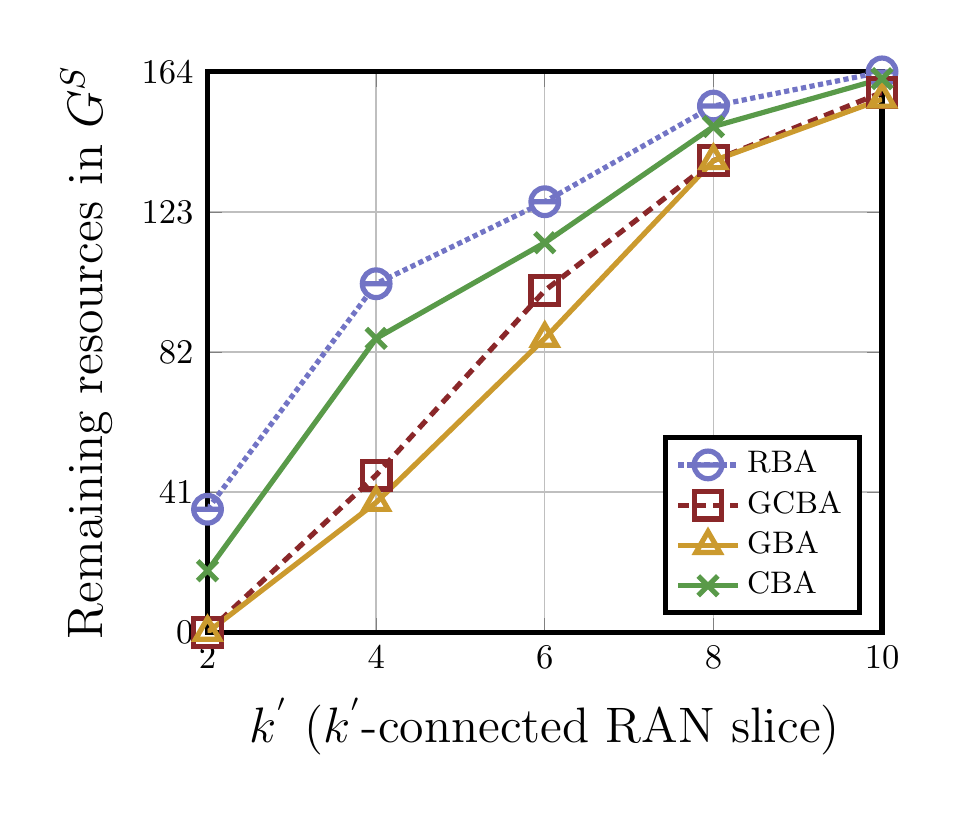} \label{fig:r-degvnf-recava}}
\caption{\textbf{Shortage Case:} total amount of available resources in the substrate network when a) the number of substrate nodes ranging from 60 to 140, b) number of VNFs ranging from 160 to 240, c) $k$-connected substrate network with degree ($k$) ranging from 2 to 10 per substrate node, and d) $k'$-connected RAN slice with degree ($k'$) ranging from 2 to 10 per VNF.}
\vspace{-15pt}
\label{fig:r-resource}
\end{figure}

In this section we conduct simulations to evaluate the algorithms in terms of the resource utilization efficiency in both scenarios (normal case in Fig. \ref{fig:a-resource} and shortage case in Fig. \ref{fig:r-resource}).
Fig. \ref{fig:a-subnode-numrecava} shows comparisons of the amount of remaining resources of the substrate network ($G^S$) after finishing the embedding process when the number of substrate nodes increasing from $60$ to $140$. The figure illustrates that RBA and CBA have a lower amount of remaining resources while committing a higher embedding performance as demonstrated in Fig. \ref{fig:a-embedding}.
As explained in the Fig. \ref{fig:a-embedding} and Fig. \ref{fig:r-embedding}, both RBA and CBA consider the cluster of VNFs and substrate nodes during the embedding process that helps leverage of not only embedding performance but also resource utilization performance. In addition, CBA and RBA provide a fair performance in terms of resource utilization with a lower time complexity.

Fig. \ref{fig:a-vnf-num-recava} shows the impact of remaining resources of the substrate node ($G^s$) with respect to the number of VNFs ranging from $160$ to $240$. As expected, GBA and GCBA continue to show better embedding performance as seen in Fig. \ref{fig:a-embedding} while still consuming minimal amount of resources. In contrast to, CBA and RBA, which have a slight decrease in the embedding efficiency even though, both algorithms utilize high resources to embed VNFs to the substrate network.

Fig. \ref{fig:a-degsubnode-recava} shows comparison of the remaining resources of the substrate node ($G^S$) with respect to $k$ where $k$ denotes the degree of each substrate node in the substrate network, ranging between $2$ to $10$. In the figure, GBA and GCBA still show superior performance invariably compare to CBA and RBA. Moreover, Fig. \ref{fig:a-degsubnode-recava} illustrates that as the degree of substrate node increases, the number of embedded VNFs increases in the algorithms GBA and GCBA, with still having a high amount of resources remaining, on the other hand RBA and CBA consume high amount of resources, but still result in lower performance efficiently.

Fig. \ref{fig:a-degvnf-recava} we analyze the comparisons of the proposed algorithms with different $k^{'}$ with respect to the remaining resources of each substrate node , which $k^{'}$ denotes the degree of VNFs in RAN slices ranging from 2 to 10. GBA and GCBA continue to perform better when comparing with CBA and RBA. As shown in Fig. \ref{fig:a-degvnf-recava} when the degree per substrate node is $2$ GBA and GCBA embed VNFs with as lower as $204$ remaining resources, but this not the case with RBA and CBA as both algorithms even after having to consume $254$ total resources still embed VNFs at a lower efficiently with comparison to GBA and GCBA.

Fig. \ref{fig:a-resource} shows comparisons of the amount of remaining resources of the substrate network ($G^S$) after finishing the embedding process when the number of substrate nodes increasing from $60$ to $140$. The figure illustrates that GBA and GCBA have a higher amount of remaining resources while committing a higher embedding performance as demonstrated in Fig. \ref{fig:a-embedding}.
As explained in the Fig. \ref{fig:a-embedding} and Fig. \ref{fig:r-embedding}, both GBA and GCBA consider the cluster of VNFs and substrate nodes during the embedding process that helps leverage of not only embedding performance but also resource utilization performance.
In addition, CBA and RBA provide a fair performance in terms of resource utilization with a lower time complexity.

Fig. \ref{fig:r-subnode-recava} illustrates comparisons of the amount of remaining resources of the substrate network ($G^S$) after completing the embedding process when the number of substrate nodes are increased from $60$ to $140$. The figure demonstrates that GBA and GCBA have a higher amount of remaining resources while committing a higher embedding performance efficiency as shown in Fig. \ref{fig:r-subnode-recava}.
As described in the Fig. \ref{fig:a-embedding} and Fig. \ref{fig:r-embedding}, both GBA and GCBA during the embedding process consider not only the required resources of a single VNF but also take into consideration of the required resources of a cluster of VNFs. GBA and GCBA, therefore, provide a better performance in terms of resource utilization with a fair time complexity.

In the Fig. \ref{fig:r-vnf-numrecava} we analyze comparisons of the algorithms capability in terms of the remaining resources in the substrate node ($G^S$) with respect to the number of VNFs ranging from $160$ to $240$. As expected both GBA and GCBA still continue to show better performance efficiency while embedding the VNFs with high remaining resources. In contrast to CBA and RBA performing slightly towards the lower end in terms of the embedded VNFs, on the other hand consuming high resources to embed the VNFs.

Fig. \ref{fig:r-degsubnode-recava} shows the performance of the proposed algorithms when varying the amount of remaining resources in the substrate nodes ($G^S$) with respect to $K$ which denotes the degree of the substrate node in the substrate network from $60$ to $140$. The performance of GBA and GCBA still shows a higher result, when compare to both CBA and RBA having a lower performance efficiency. Moreover, we can see that the higher the amount of remaining resources in the substrate network, higher is the embedding performance of GBA and GCBA while having the amount of remaining resources high.

Fig. \ref{fig:r-degvnf-recava} shows comparisons of the remaining resources in the substrate network ($G^S$) with respect to $k^{'}$, which denotes the degree of VNFs in an RAN slice, ranging from $2$ to $10$. GBA and GCBA show a higher total number of embedded VNFs in the substrate network. In contrast, RBA and CBA have a lower number of total embedded VNFs. As expected by this mechanism, GBA and GCBA embed more VNFs with still having high remaining resources, on the other hand CBA and RBA while taking more resources to embed have a lower embedding efficiency.
In summary, Fig. \ref{fig:r-resource} evaluations show that the comparisons of shortage case in Fig. \ref{fig:r-subnode-recava} -- Fig. \ref{fig:r-degvnf-recava} in terms of amount of remaining resources in the substrate network($G^S$) when varying the substrate nodes, VNFs, $k$, and $k^{'}$. GBA and GCBA consistently provide a better performance in terms of amount of resource utilization than RBA and CBA.

\section{Conclusion} \label{sec:conclusion}

In this paper we have established the theoretical foundation for using RS-configuration to construct a VNF embedding plan for the RAN slicing. We have formulated the efficient resource allocation as an essential problem with the objective to maximize the total number of embedded VNFs.
To solve the RS-configuration problem, we have introduced four efficient algorithms (RBA, CBA, GCBA and GBA) and show theoretical analyses to demonstrate the efficacy of the algorithms.
Extensive simulation results have been provided to evaluate the performance of the proposed algorithms using different metrics in terms of the embedding performance and resource utilization performance. We have created different scenarios to test the algorithms, including considering the network under the normal condition (normal case) and under the resource shortage condition (shortage case). Through the results of the simulations, GBA and GCBA consistently demonstrate a better performance for both embedding VNFs and utilizing the network's resources when comparing to RBA and CBA.
However, in terms of the time complexity RBA and CBA demonstrate as faster algorithms with lower time complexity while remaining fair embedding VNFs and resource utilization performances.

\section*{ACKNOWLEDGEMENT}
This research was supported in part by the US NSF grant CNS-2103405.

\normalem

\bibliographystyle{IEEEtran}
\bibliography{bib-tn}
\end{document}